\documentclass[11pt]{article}
\usepackage[margin=1in]{geometry}
\usepackage{amssymb,amsmath,amsthm}
\usepackage{graphicx, graphics, color}
\usepackage{multirow}
\usepackage{algorithm}
\usepackage{algpseudocode}
\usepackage{authblk}
\usepackage{epstopdf}
\usepackage{natbib}
\usepackage[colorlinks=true, linkcolor=blue, citecolor=blue]{hyperref}
\usepackage{url}
\usepackage{bm}
\usepackage{setspace} 
\usepackage{comment}

\newcommand{\diag}{\mbox{diag}}
\newcommand{\cov}{\mbox{cov}}
\newcommand{\cor}{\mbox{cor}}
\newcommand{\var}{\mathrm{var}}
\newcommand{\tr}{\mbox{tr}}

\newtheorem{theorem}{Theorem}
\newtheorem{lem}{Lemma}
\newtheorem{corollary}{Corollary}
\newtheorem{proposition}{Proposition}

\title{Locally Optimal Design for A/B Tests in the Presence of Covariates and Network Dependence}
\author[1]{Qiong Zhang\thanks{qiongz@clemson.edu}}
\author[2]{Lulu Kang\thanks{lkang2@iit.edu}}
\affil[1]{School of Mathematical and Statistical Sciences, Clemson University}
\affil[2]{Department of Applied Mathematics, Illinois Institute of Technology}
\date{}  

\begin{document}
\maketitle

{\singlespacing
\begin{abstract}
A/B test, a simple type of controlled experiment, refers to the statistical procedure of experimenting to compare two treatments applied to test subjects. 
For example, many IT companies frequently conduct A/B tests on their users who are connected and form social networks. 
Often, the users' responses could be related to the network connection. 
In this paper, we assume that the users, or the test subjects of the experiments, are connected on an undirected network, and the responses of two connected users are correlated. 
We include the treatment assignment, covariate features, and network connection in a conditional autoregressive model. 
Based on this model, we propose a design criterion that measures the variance of the estimated treatment effect and allocate the treatment settings to the test subjects by minimizing the criterion. 
Since the design criterion depends on an unknown network correlation parameter, we adopt the locally optimal design method and develop a hybrid optimization approach to obtain the optimal design. 
Through synthetic and real social network examples, we demonstrate the value of including network dependence in designing A/B experiments and validate that the proposed locally optimal design is robust to the choices of parameters. 
\end{abstract}

}


\section{Introduction}\label{sec:intro}

A/B or A/B/n test, a simple type of controlled experiment, refers to the procedure of comparing the outcomes of two or more treatment settings from a finite number of test subjects. In the literature, controlled experiments
have been widely used in agricultural, clinical trials, engineering and science studies, marketing research, etc \citep{atkinson2001one}.
Due to the advent of Internet technologies, large-scale A/B test has been commonly used by technology companies such as Amazon, Facebook, LinkedIn, Netflix, etc., to compare different versions of algorithms, web designs, and other online products and services.
For example, \cite{nandy2020b} showed a case study on the LinkedIn newsfeed, which is a content recommender system with hundreds of millions of users. 
A recommender system is referred to as the infrastructure that provides a personalized recommendation on products or services based on users' personal information or past behaviors \citep{kohavi2020trustworthy}.
The accuracy of the recommender algorithm is crucial to the quality and/or profit of these companies. 
A practical problem is to decide if an innovative update should be made to the algorithm in use. 
Therefore, an A/B test (i.e., ``A'' refers to the updated algorithm and ``B'' the current one) is used to make a comparison of the two and make the decision. 
To make a robust comparison of the two algorithms, the experiment should last for a certain period to make sure that users can receive enough exposure to the updated recommender system. 
The outcome of each user can be the total time spent on the recommended products/service and the click-through rate to the recommended products/service. 

In its simplest form, the experimenter wants to compare the outcomes of two different treatments, labeled by A and B.
A completely randomized design is commonly used, in which the treatment setting is randomly assigned to different test subjects.
The randomization leads to unbiased estimates of certain estimands, typically, the average treatment (or causal) effect \citep{rubin2005causal}, under minimum assumptions.
However, there is still room for improvement in the efficiency of the A/B test procedure when certain practical challenges are involved.
Besides the treatment setting, many other variables can affect a user's outcome, including the covariates information and social network connection of the user. 
Covariates, such as users' demographic, educational, financial information, are usually available to the experimenter and can significantly contribute to the behaviors and opinions of users. 
In the aforementioned scenarios, the experimenter also possesses the network connections of the users.
In Section \ref{sec:real}, we simulate an A/B experiment for the music recommender system based on the real dataset collected from the music streaming service Deezer. 
The data contains the friendship network of users and their covariates information regarding preferences to different music genres, which should be highly influential to the music recommender system. 
Intuitively, the outcomes of two connected users might be correlated to some degree.
This intuition is reflected in the model assumption of the outcome regarding the network structure, and referred to as the network-correlated outcomes in \cite{basse2018model}.
In Section \ref{sec:literature} and \ref{sec:optd}, we explain in details the assumption on the effects of the network to a user's outcome.

The rest of the paper is arranged as follows. 
In Section \ref{sec:literature} we highlight some relevant existing works and point out the differences between the proposed method and the existing ones. 
Section \ref{sec:optd} introduces the regression model including both the covariates and the correlation between users due to network connections. 
Based on this model, in Section \ref{sec:BayOptD}, we propose a locally optimal design method in which the network correlation parameter is set to be the mean of its prior distribution.
In Section \ref{sec:hybrid}, a hybrid approach is proposed to solve the optimization problem to obtain the optimal design.
Through numerical experiments in Section \ref{sec:num} and \ref{sec:real}, we demonstrate the benefit of the proposed approach.
We conclude the paper in Section \ref{sec:con} with some discussion of the limitation of the proposed method and some future research directions.

\section{Previous Work and Our Contribution}\label{sec:literature}

\subsection{Existing Literature}\label{sec:review}

For A/B tests that only involve covariates but not networks, most existing works advocate the necessity of covariate balancing between the treatment groups \citep{morgan2012rerandomization, rubin2005causal, morgan2015rerandomization, bertsimas2015power, kallus2018optimal,li2021covariate}. 
For controlled experiments on networks, both theoretical and methodological works have been developed.
See \cite{gui2015network, phan2015natural, eckles2015design, basse2018limitations}, etc.
Among them, \cite{gui2015network} proposed an estimator of average treatment effect considering the interference between users on the network and a randomized balance graph partition to assign treatments to each of the subnetworks.
\cite{eckles2015design} used a graph cluster randomization to reduce the bias of the average treatment effect estimate.
\cite{nandy2020b} proposed the strategy to first apply approximate randomized controlled experiments solved by optimization and then use importance sampling to correct bias. 
Although focusing on networks, these works do not consider covariates.

In causal inference literature, the potential outcome framework is usually used. 
The average treatment effect is the target parameter for estimation and inference \citep{imbens2015causal}.
Under this setup, many causal inference works do not require any probabilistic model assumption on the response variable.
Alternatively, some recent works on the design for A/B experiments have operated under specific parametric model assumptions of the response variable, and optimal design idea is used to propose new design methods.
For example, \cite{bhat2020near} developed off-line and online mathematical programming approaches to solve this optimization problem, the objective function of which is exactly the $D_s$-optimal design criterion \citep{kiefer1961optimum,atkinson1992optimum}.
In this case, the $D_s$-optimality criterion minimizes the variance of the treatment effect of a parametric linear model.
Optimal design strategies have also been used under the assumption of the network-correlated outcome, such as \cite{basse2018model} and \cite{pokhilko2019d}. 
Outside the A/B test literature, there have been papers considering the optimal design problem with dependence between test subjects.
For example, \cite{martin1986design} considered the restricted randomized design when the test subjects are spatially correlated.
\cite{parker2017optimal} and \cite{koutra2017designing} considered the optimal design under linear network effects.

Similar to the aforementioned works, we also opt for the optimal design direction as indicated by the title of this paper. 
We argue that although the nonparametric potential outcome framework has an important theoretical basis, the reasonable model-assisted design approaches are not meritless. 
Even in the works based on the potential outcome framework, certain linear model assumptions are also used in both theoretical and numerical proofs to show the advantages and properties of the balancing criteria and the design approaches.
For example, \cite{morgan2012rerandomization} assumed an additive linear model to show how much variance reduction can be obtained by rerandomization using Mahalanobis distance.
\cite{gui2015network} used a linear additive model in terms of treatment effect, neighboring covariates, and neighboring responses as the rationale to create the sample estimator of average treatment effect, as well as to simulate data in numerical experiments.

\subsection{Differences and New Contributions}\label{sec:diff}

In this paper, we develop an optimal design approach for A/B experiments in the presence of both covariates and network connections.
The scope of the paper targets the social networks of users whose covariates information are influential to their reactions to the treatments. 
With a parametric conditional autoregressive (CAR) model that assumes the outcome is the sum of treatment effect, covariate effects, and correlated residuals for capturing network dependence, we focus on the estimation of the treatment effect parameter.
Based on this model, we develop an optimal design criterion such that the variance of the estimated treatment effect is minimized.
By design, we mean the assignment of treatment settings to each test subject in the context of this paper.
We focus on the simplest case where the experiment only involves two treatments, A and B.
But the proposed modeling and design method can be extended to the case of multiple treatment settings, as discussed in Section \ref{sec:con}.
The design of the treatment settings for multiple experimental factors is not the focus of this paper.

The resulting design criterion in Section \ref{sec:optd} depends on the network structure, the covariates, and an unknown network correlation parameter, and it can not be simply expressed as a sparse quadratic function of the design variables, which is different from \cite{pokhilko2019d}. 
Therefore, the mathematical formulation developed by \cite{pokhilko2019d} is infeasible to solve this new optimal design problem. 

We also assume the common Stable Unit Treatment Value Assumption (SUTVA) \citep{rubin1974estimating}, which states that the outcome of a test subject is unaffected by the treatment assignments of any other subjects.
In other words, we do not think there is any direct interference from the neighbors' treatment settings to the focused test subject's outcome. 
This assumption is appropriate for many applications where users are unawarely participating in the experiments run by the online service providers. 
Users' outcomes can still be correlated due to their network connections and covariates information. 

This non-interference assumption is different from the interference assumption in some existing works, such as \cite{parker2017optimal}. 
In \cite{parker2017optimal}, the proposed model includes the treatment assignments of connected subjects as linear predictors in their model. 
The experimental outcome of a subject under this model is affected by the treatment assignments of connected subjects. Different from this assumption, we assume that the experimental outcomes are correlated due to the network connection between subjects, which is characterized by the error term of the CAR model. However, the experimental outcomes are unaffected by the treatment assignments of connected subjects.  
Therefore, the proposed model of this paper is not comparable with the one in \cite{parker2017optimal} 
due to the different assumptions. 
Both can be useful under suitable scenarios and assumptions.

\section{Optimal Design with Network Connection}\label{sec:optd}

Consider $n$ test subjects participating in the experiment.
For the $i$-th subject, let $x_i\in\{-1,1\}$ represent the experimental allocation of A or B treatment, $\bm z_i=(z_{i1},\ldots,z_{ip})^\top$ be the $p$-dimensional covariates, and $y_i$ be the experimental outcome.
Assume that the outcome $y_i$ is a continuous random variable.
\cite{bhat2020near} models the relationship between $y_i$ and the effects of the treatment and covariates as
\begin{equation}\label{eq:lm}
  y_i=x_i\theta + \bm f^\top_i\bm \beta+\delta_i~~\mathrm{for}~~i=1, \ldots, n,
\end{equation}
Here $\bm \beta\in \mathbb{R}^{p+1}$ is the vector of the linear coefficients for $\bm f_i=(1, \bm z^\top_i)^\top$.
We name $\theta$ as the treatment effect.
Note that it is different from the notion of average treatment effect
which is the usual estimand in the potential outcome framework.
The model in \eqref{eq:lm} does not involve the network and the error terms $\delta_i$'s are assumed to be independent and identically distributed (iid) normal random variables with mean zero and a constant variance $\sigma^2$.
The purpose of design allocation is to reduce the variance of the least square estimator, which is unbiased if assumption \eqref{eq:lm} stands.
According to \cite{bhat2020near}, the variance of the least squares estimator $\hat\theta$ from \eqref{eq:lm} is $\mathrm{var}(\hat\theta)=\sigma^2[\bm x^\top (\bm I_n-\bm F(\bm F^\top \bm F)^{-1}\bm F^\top)\bm x]^{-1}$,
where $\bm x=(x_1, \ldots, x_n)^\top$, $\bm F^\top =(\bm f_1, \ldots, \bm f_n)$ and $\bm I_n$ is the identity matrix of size $n$. Therefore, the optimal design is obtained by minimizing $\var(\hat{\theta})$, which is equivalent to
\begin{align}\label{eq:opt_no}
\min & ~\bm x^\top \bm F(\bm F^\top \bm F)^{-1}\bm F^\top \bm x\\
\text{s.t. }   &-1\leq \sum^n_{i=1} x_i\leq 1,\quad \bm x \in \{-1,1\}^n. \nonumber
 \end{align}
The constraint $-1\leq \sum_{i=1}^n x_i\leq 1$ is imposed to make sure that the numbers of test subjects assigned to 1 and $-1$ are equal or within the difference of 1, which corresponds to the situation of even or odd sample size $n$.

Next, we extend the linear model \eqref{eq:lm} to the case with the network connection.
We require that the network between subjects is known to the experimenter just like the covariate information.
Assume this network form a simple undirected graph with nodes representing the test subjects.
If two test subjects are connected, there is a single edge between the two corresponding nodes.
Such a network can be represented by an $n\times n$ adjacency matrix $\bm W$, or incidence matrix.
Its diagonal entries are 0's, whereas off-diagonal entries ($i\neq j$) are
\begin{equation}
\label{eq:W}
w_{ij}=\begin{cases}
1, & \text{if node $i$ and node $j$ are adjacent or connected}\\
0, & \text{otherwise}.
\end{cases}
\end{equation}
Obviously, $\bm W$ is symmetric.
We denote the number of adjacent neighbors, or degree, of the $i$-th node as $m_i=\sum^n_{j=1}w_{ij}$, and $m=\sum^n_{i=1}m_i$ is twice of the total number of edges in this graph.

To add the network's influence into the linear additive model \eqref{eq:lm}, we propose the conditional autoregressive or CAR distribution \citep{cressie1993spatial, rue2005gaussian, banerjee2014hierarchical} for $\delta_i$'s to represent the network dependence between the connected test subjects. 
According to the CAR model, 
\begin{equation}\label{eq:delta}
\delta_i|\delta_1, \ldots, \delta_{i-1}, \delta_{i+1}, \ldots, \delta_n\sim N\left(\rho\sum_{j\neq i}\frac{w_{ij}\delta_j}{m_i}, \frac{\sigma^2}{m_i}\right),
\end{equation}
where $\sigma^2$ is the variance and $0\leq\rho< 1$ is a correlation parameter characterizing the strength of network dependence. 
Equivalently, $\bm\delta=(\delta_1, \ldots, \delta_n)^\top$ follows a multivariate normal distribution
\begin{equation}\label{eq:car1}
\bm\delta\sim \mathcal {MVN}_n(0, \sigma^2\bm R^{-1}(\rho, \bm W)),
\end{equation}
where $\bm R(\rho, \bm W)=(\bm D-\rho \bm W)$ with $\bm D=\diag\{m_1, \ldots, m_n\}$.
The matrix $\bm R(\rho, \bm W)$ is positive definite when $0\leq \rho < 1$ and $m_i\geq 1$ for $i=1, \ldots, n$ \citep{ver2018relationship}.
The proof of the equivalence of \eqref{eq:delta} and \eqref{eq:car1} is given by \cite{besag1974spatial} and illustrated by \cite{pokhilko2019d} under the framework of network A/B test.

With the CAR model assumption, the outcome $y_i$ depends on the network connection for the $i$-th subject but does not depend on the treatment allocation of the connected subjects.
In this way, the outcome of the subject is mainly decided by him/herself and the treatment he/she receives.
Since most social networks are built on positive connections between users, we assume the influence from the network is synergistic to users and thus $\rho$ is positive.
When $\rho=0$, the model assumption returns to the linear model \eqref{eq:lm} which does not involve a network. 

If the network correlation parameter $\rho$ is known and $0\leq \rho <1$, the variance of the least squares estimator $\hat\theta$ can be expressed by
\begin{equation}\label{eq:var_theta}
\var(\hat{\theta})=\sigma^2 \left[{\bm x}^\top \bm K\bm x\right]^{-1},
\end{equation}
where $\bm x=(x_1, \ldots, x_n)^\top$ is the treatment assignments for all subjects, and $\bm K$ is an $n\times n$ matrix
\begin{equation}\label{eq:K}
\bm K=(\bm D-\rho \bm W)-(\bm D-\rho \bm W)\bm F\left[\bm F^\top (\bm D-\rho \bm W) \bm F\right]^{-1}\bm F^\top (\bm D-\rho \bm W),
\end{equation}
with the covariates matrix $\bm F^\top =(\bm f_1, \ldots, \bm f_n)$.
Here, to make $\bm F^\top (\bm D-\rho \bm W) \bm F$ invertible, we require $\bm F$ to be full-rank, i.e., $\text{rank}(\bm F)=p+1$.
For the same reason, we require $m_i\geq 1$ for all the nodes. 
Therefore, the proposed CAR model does not apply to the isolated nodes whose degree $m_i=0$.
In Section \ref{sec:real}, we explain how to deal with the isolated nodes if they exist.
The optimal design $\bm x$ minimizes the variance of estimated treatment effect in \eqref{eq:var_theta}.
This optimal design is also known as a $D_s$-optimal design in the optimal design literature \citep{kiefer1961optimum,atkinson1992optimum}.
Equivalently, we express the optimal design as the solution of
\begin{align}\label{eq:max}
  \mathrm{max} ~& T(\bm x, \rho):=\bm x^\top \bm K\bm x,\\
  \text{s.t. } -1\leq &\sum_{i=1}^n x_i\leq 1, \text{ and } \bm x \in \{-1,1\}^n,\nonumber
\end{align}
which maximizes the precision (as the inverse of variance) of the estimated treatment effect.
Since
\begin{equation}\label{eq:objK}
\bm x^\top \bm K\bm x=\bm x^\top(\bm D-\rho \bm W)\bm x-\bm x^\top(\bm D-\rho \bm W)\bm F\left[\bm F^\top (\bm D-\rho \bm W)\bm F\right]^{-1}\bm F^\top (\bm D-\rho \bm W)\bm x,
\end{equation}
and $\bm x^\top(\bm D-\rho \bm W)\bm x=m-\rho \bm x^\top\bm W\bm x$, we have that
\begin{equation}\label{eq:min}
T(\bm x, \rho)=m-T_1(\bm x, \rho)-T_2(\bm x, \rho),
\end{equation}
where
\begin{align*}
T_1(\bm x, \rho)& =\rho\bm x^\top \bm W\bm x=\rho\sum_{ij}w_{ij}x_i x_j, \\
T_2(\bm x, \rho)&=\bm x^\top (\bm D-\rho \bm W)\bm F\left[\bm F^\top (\bm D-\rho \bm W)\bm F\right]^{-1}\bm F^\top (\bm D-\rho \bm W)\bm x.
\end{align*}
Note that minimizing $T_1(\bm x, \rho)$ would push $x_i$ and $x_j$ to be assigned with different treatments whenever $w_{ij}=1$.
To facilitate the discussion, we name this condition ``connection balance'', meaning that the two connected subjects are assigned with different treatment settings.
Intuitively, this is a meaningful condition since two connected test subjects are usually similar in many aspects of their background.
Thus, the most likely factor contributing to their difference in outcomes is the treatment setting.
This condition is consistent with the optimal design for the A/B test without covariates in \cite{pokhilko2019d}.
In the extremely simple and artificial case illustrated later in Figure \ref{fg:illustration} in Section \ref{sec:hybrid}, such perfect balance can be achieved. 
For real networks, the connection balance can only be achieved to a certain degree but rarely perfectly. 
Also, $T_2(\bm x, \rho)$ can be viewed as a network re-weighted Mahalanobis distance in \cite{morgan2012rerandomization}, since it can be expressed by
\[
T_2(\bm x, \rho) =\bm x^\top (\bm D-\rho \bm W)\bm F\bm \Sigma^{-1}_n \bm F^\top (\bm D-\rho \bm W)\bm x,
\]
with $\bm \Sigma_n=\bm F^\top (\bm D-\rho \bm W)\bm F$.
Therefore, the objective in \eqref{eq:min} contains $T_1(\bm x, \rho)$ to achieve connection balance, and $T_2(\bm x, \rho)$ to achieve covariate balance.
One critical issue is that the optimality criterion depends on the value of $\rho$.
In practice, $\rho$ is an unknown parameter.
Next, we are going to discuss the choice of $\rho$.

\section{Locally Optimal Design}\label{sec:BayOptD}

The optimal design criterion $T(\bm x,\rho)$ depends on the network correlation parameter $\rho$, which is usually unknown before experiments.
We can use Bayesian optimal design to handle the uncertainty of the unknown parameters.
Using its most common formulation, we should optimize the expectation of the design criterion, i.e., $\mathbb{E}_{\rho}[T(\bm x, \rho)]$, with respect to a user-specified prior distribution of the parameter $\rho$.
But even with the simple uniform prior for $\rho$, the expectation does not have a tractable form.
Many numerical methods, such as quadrature, Quasi-Monte Carlo, Markov Chain Monte Carlo, etc., have to be used to compute the integration.
Please see \cite{ryan2014towards}, \cite{ryan2016review}, and \cite{drovandi2018improving} for more comprehensive review on the advanced computational methods on Bayesian optimal designs.

To simplify the computation, we investigate the property of $T(\bm x, \rho)$ with respect to $\rho$ to find an analytic surrogate of $\mathbb{E}[T(\bm x, \rho)]$.
We first discover the concavity of $T(\bm x, \rho)$ with respect to $\rho$ in Theorem \ref{thm:concave}.
Based on Jenson's Inequality, the conclusion in Corollary \ref{cor:upper} holds directly.
The proof of Theorem \ref{thm:concave} is provided in the Supplement.
Based on the two results, we propose to use $T(\bm x,\rho_0)$, the upper bound of $\mathbb{E}[T(\bm x, \rho)]$, as the surrogate of the objective to obtain design allocation.

\begin{theorem}\label{thm:concave}
For $\rho \in (0,1)$, and any given design $\bm x$, the design criterion $T(\bm x,\rho)$ is a concave function with respect to $\rho$.
\end{theorem}
\begin{corollary}\label{cor:upper}
Given a prior distribution of $\rho$, $p(\rho)$, for $\rho\in (0,1)$, a tight upper bound for $\mathbb{E}\left[T(\bm x, \rho)\right]$ is $\mathbb{E}\left[T(\bm x, \rho)\right]\leq T(\bm x,\rho_0)$, where $\rho_0:=\mathbb{E}(\rho)$ is the population mean of $\rho$ based on $p(\rho)$.
\end{corollary}
We define the locally optimal design by solving
\begin{align}\label{eq:opt_bayesian}
\max & ~T(\bm x, \rho_0)\\
\text{s.t. }   -1\leq &\sum_{i=1}^n x_i\leq 1, \text{ and } \bm x \in \{-1,1\}^n,\nonumber
\end{align}
whose objective function is equivalent to the original objective in \eqref{eq:min} with $\rho$ specified as the mean of the prior distribution.
Using a specific $\rho_0$ in the design criterion to obtain the optimal design is known as the locally optimal design \citep{chaloner1995bayesian}.

The quality of the design based on the surrogate problem in \eqref{eq:opt_bayesian} can be investigated from two aspects. 
First, we provide the analytic gap between $T(\bm x, \rho_0)$ and $\mathbb{E}[T(\bm x, \rho)]$ and a simulation example to illustrate the typical range of the gap between the surrogate local design criterion $T(\bm x, \rho_0)$ and the global criterion $\mathbb{E}[T(\bm x,\rho)]$.
Proposition \ref{prop:gap} of the analytic gap and simulation results in Figure \ref{fg:gap} are given in the Supplement.

Next, we investigate whether the surrogate design criterion $T(\bm x,\rho_0)$ is robust to the choice of $\rho_0$. 
To do so, we check of the correlation between any $T(\bm x,\rho_0)$ and $T(\bm x,\rho)$ for a pair of fixed $(\rho_0, \rho)$ for any randomly generated design $\bm x$. 
If the correlation between  $T(\bm x,\rho_0)$ and $T(\bm x,\rho)$ is large and positive, it indicates that a design resulting in large $T(\bm x,\rho_0)$ is also likely to lead to large $T(\bm x,\rho)$. 
In Proposition \ref{prop:cor}, we have given the formula to calculate the exact correlation $\cor_{\bm x}(T(\bm x,\rho_0),T(\bm x,\rho))$ for all the completely randomized design in which $x_i$'s are i.i.d. random variables and $\Pr(x_i=1)=\Pr(x_i=-1)=0.5$. To visualize the correlation,
we also provide a simulated example using a network with 50 nodes and five-dimensional covariates associated with each node.
The edges of the network are generated as independent Bernoulli random variables with a probability of 0.08. The covariates are generated as independent random variables taking values from  $\{-1, 1\}$ with equal probabilities.
The values of $\rho_0$ and $\rho$ are set to be 0.1, 0.3, 0.5, 0.7, and 0.9 and omit the case when $\rho_0=\rho$. 
For each pair of $(\rho_0, \rho)$ values, we generate 1000 completely randomized designs and compute the corresponding $T(\bm x,\rho_0)$ and $T(\bm x,\rho)$ for each design. 
Figure \ref{fig:cor} returns the scatter plot of $T(\bm x,\rho_0)$ and $T(\bm x,\rho)$ for the 1000 completely randomized design for different $(\rho_0, \rho)$ values. 
It shows that $T(\bm x, \rho_0)$ and $T(\bm x, \rho)$ are strongly linearly correlated.
Also, the \emph{exact} correlation values based on Proposition \ref{prop:cor} ranges from  0.75-0.99 for values of $(\rho_0,\rho)$ in the simulation.
From these results, it is safe to say that the locally optimal design is robust to the choice of $\rho_0$ value. 
Particularly, for $\rho_0=0.5$, the correlation values between $T(\bm x,\rho_0)$ and $T(\bm x, \rho)$ where $\rho=0.1, 0.3, 0.7, 0.9$ are all above 0.9. 
Therefore, the optimal design obtained based on $\rho_0=0.5$ is the most robust for the model with the true value of $\rho\in (0.1, 0.9)$.

\begin{figure}
	\centering
	\includegraphics[scale=0.6]{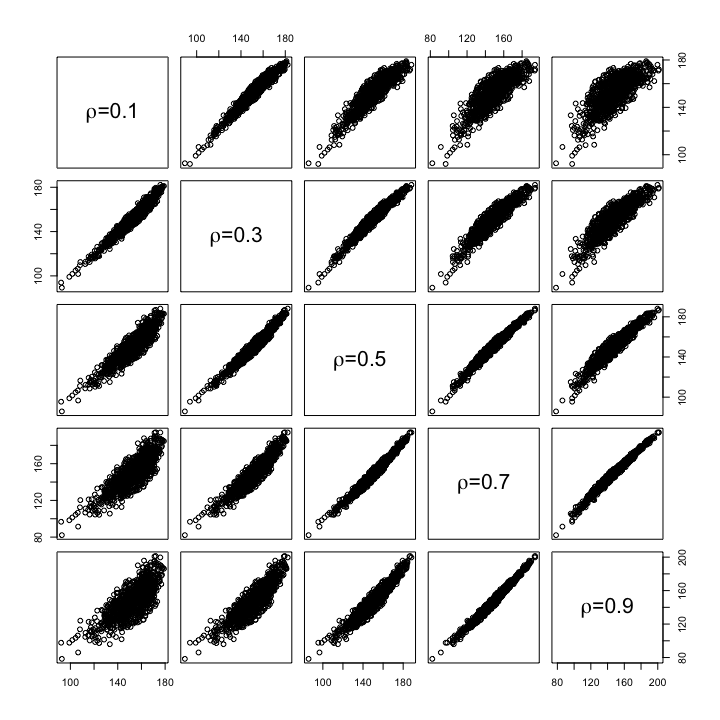}
	\caption{Scatter plot of $T(\bm x, \rho_0)$ and $T(\bm x, \rho)$  of each pair of $(\rho_0, \rho)$ with 1000 randomly generated designs.}\label{fig:cor}
\end{figure}

Although using the local design criterion $T(\bm x,\rho_0)$ is a simple solution, the quality of the resulting design can be validated. 
Simulation examples in Section \ref{sec:num} can further demonstrate that the performance of the locally optimal design is equally good as the \emph{true} optimal design in which parameter $\rho$ is set to be its true known value.

\section{A Hybrid Solution Approach to Obtain Optimal Design}\label{sec:hybrid}

Since the optimal design in \eqref{eq:opt_bayesian} is the integer solution of the maximum of a quadratic form, obtaining the exact solution of such problem is challenging \citep{belotti2013mixed, bhat2020near}.
According to \eqref{eq:min} and \eqref{eq:opt_bayesian}, the maximization can be converted to minimization of $T_1(\bm x, \rho_0)+T_2(\bm x, \rho_0)$, where
\begin{align*}
T_1(\bm x, \rho_0)& =\rho_0\bm x^\top \bm W\bm x\\
T_2(\bm x, \rho_0)&=\bm x^\top (\bm D-\rho_0\bm W)\bm F\left[\bm F^\top (\bm D-\rho_0\bm W)\bm F\right]^{-1}\bm F^\top (\bm D-\rho_0\bm W)\bm x,
\end{align*}
and 
\[
T_1(\bm x, \rho_0)+T_2(\bm x, \rho_0) =\bm x^\top \left[\rho_0\bm W+(\bm D-\rho_0\bm W)\bm F\left[\bm F^\top (\bm D-\rho_0\bm W)\bm F\right]^{-1}\bm F^\top (\bm D-\rho_0\bm W)\right] \bm x.
\]
The matrix
\[\rho_0\bm W+(\bm D-\rho_0\bm W)\bm F\left[\bm F^\top (\bm D-\rho_0\bm W) \bm F\right]^{-1}\bm F^\top (\bm D-\rho_0\bm W)\]
is not necessarily a positive semi-definite matrix. 
As a result, the minimization problem of $T_1(\bm x,\rho_0)+T_2(\bm x,\rho_0)$ can not be solved directly as the problem in \eqref{eq:opt_no}.
We develop a hybrid solution approach to resolve this issue.

Notice that the matrix $(\bm D-\rho_0\bm W)\bm F\left[\bm F^\top (\bm D-\rho_0\bm W)\bm F\right]^{-1}\bm F^\top (\bm D-\rho_0\bm W)$ in $T_2(\bm x, \rho_0)$ is positive definite, and thus the minimization of $T_2(\bm x, \rho)$ can be solved by an outer-approximation based branch-and-cut algorithm as in \eqref{eq:opt_no}. 
Hence, we reformulate the minimization of $T_1(\bm x, \rho_0)+T_2(\bm x, \rho_0)$ to be
\begin{align} \label{eq:formulation_temp}
   \min \ &     T_2(\bm x, \rho_0)\\
   \text{s.t. } &  T_1(\bm x, \rho_0) \leq q, \nonumber\\
    -1\leq &\sum_{i=1}^n x_i\leq 1, \text{ and } \bm x \in \{-1,1\}^n,\nonumber
\end{align}
which uses the constraint $T_1(\bm x, \rho_0) \leq q$ to control the value of $T_1(\bm x, \rho_0)$ to be small enough.
Since $\rho_0$ is only a constant multiplier in $T_1(\bm x,\rho_0)$, this constraint can be reduced to $\bm x^\top \bm W \bm x\leq q$, and it is critical to specify the value of $q$.

The value of $\bm x^\top \bm W \bm x$ greatly depends on the number of subjects and the structure of the network.
Therefore, the cap $q$ should be related to a specific network.
In Theorem \ref{thm:t1}, we investigate the asymptotic behaviors of $\bm x^\top \bm W \bm x$ via random allocation with equal probability to decide a viable way to specify the value of $q$.
Assume that the entire network contains unlimited users with a deterministic network structure, and the experiments are conducted on a subset of $n$ users from the entire network.
Therefore, the design vector $\bm x$ is the only random component that causes the stochastic behavior of the statistic $\bm x^\top \bm W \bm x$.
The proof of Theorem \ref{thm:t1} is provided in the Supplement. 
\begin{theorem}\label{thm:t1}
Consider that $x_1, \ldots, x_n$ in $\bm x$ are independent and identically distributed random variables from the
 discrete distribution with $\Pr(x_i=1)=\Pr(x_i=-1)=0.5$.
As $n\rightarrow\infty$,
\[
\frac{\bm x^\top \bm W \bm x}{\sqrt{m}}\xrightarrow{d} N(0, 1),
\]
where $\xrightarrow{d}$ represents convergence in distribution and $m=\sum_{i,j} w_{ij}$.
\end{theorem}

Since $\bm x^\top \bm W \bm x/\sqrt{m}$ asymptotically follows the standard normal distribution, we can specify the gap $q$ according to the standard normal percentiles. 
Let $z_\alpha$ be the $100\alpha\%$ percentile of the standard normal distribution.
If we specify a smaller value of $\alpha$, the constraint is more restrictive. 
The optimization problem is reformulated as
\begin{align} \label{eq:formulation2}
   \min      \bm x^\top (\bm D-\rho_{0}\bm W)&\bm F\left[\bm F^\top (\bm D-\rho_{0} \bm W)\bm F\right]^{-1}\bm F^\top (\bm D-\rho_0 \bm W)\bm x\\
   \text{s.t. } &  \bm x^\top \bm W\bm x \leq \sqrt{m}z_\alpha, \nonumber\\
    &-1\leq \sum^n_{i=1} x_i\leq 1,\nonumber\\
   & \bm x \in \{-1,1\}^n. \nonumber
\end{align}
This minimization problem with a positive definite quadratic objective and two-level decision variables can be solved by off-the-shelf optimization solvers.
Note that, the constraint $-1\leq \sum^n_{i=1} x_i\leq 1$ is inserted to achieve a balanced allocation of two treatments. 
In terms of implementation, this constraint usually improves the computation cost since it also reduces the number of feasible solutions.
The formulation in \eqref{eq:formulation2} changes the minimization of two objective functions $T_1(\bm x,\rho_0)$ and $T_2(\bm x, \rho_0)$ into the minimization of one and constraining the other, and thus the name of \emph{hybrid solution} approach.

We provide an illustration of the proposed design using a simple bipartite network of 20 nodes.
For simplicity, the covariate $z_i$ is a one-dimensional vector taking value from $\{-1, 1\}$. 
We set the correlation parameter $\rho_0$ as 0.5 and parameter $\alpha$ as 0.001 for the proposed approach in \eqref{eq:formulation2}.
The locally optimal design is visualized in Figures \ref{fg:illustration}.
The simple bipartite network can be divided into two disjoint sets, and the treatment allocation is orthogonal to the covariate vector, which achieves perfect balance for the covariate and the network. 
However, perfect balancing may not be achievable for general cases, but small values of $T_1$ and $T_2$ can still provide a better-balanced structure of network connection and covariates, respectively. 

\begin{figure}[t]
\centering
\includegraphics[scale=0.5, trim={100 100 100 100}]{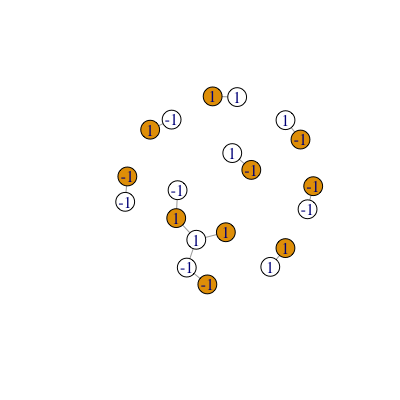}
\caption{Visualization of the optimal design allocation. Two treatments are denoted by different colors. The covariate value 1 or -1 of each subject is labeled in each node.}\label{fg:illustration}
\end{figure}

We first discuss the choice of $\alpha$ in \eqref{eq:formulation2}. The hybrid problem in \eqref{eq:formulation2} is generally computationally efficient to solve for networks with 100-5000 nodes. Therefore, it is feasible to obtain designs through conducting a sensitivity analysis with a series of decreasing $\alpha$ values and selecting the $\alpha$ value when the change of objective values $T(\bm x, \rho_0)$ in \eqref{eq:opt_bayesian} is small, or the improvement of precision stops increasing as $\alpha$ decreases. One numerical example is used to demonstrate the performance of the design to different choices of $\alpha$'s in Section \ref{sec:robust}.


At last, we remark on the choice of $\rho_0$ in this new formulation \eqref{eq:formulation2}.  
Like $T(\bm x,\rho_0)$, the new objective function $T_2(\bm x,\rho_0)$ in \eqref{eq:formulation2} is also a quadratic form of $\bm x$. 
Therefore, Proposition \ref{prop:cor} still holds for any correlation between $T_2(\bm x,\rho_0)$ and $T(\bm x,\rho)$ for any pair of $(\rho_0, \rho)$. 
The quality of design with a given $\rho_0$ can be assessed on any possible true value of $\rho$ using the analytic correlation between $T_2(\bm x, \rho_0)$ and $T_2(\bm x, \rho)$ similar to the discussion near the end of Section \ref{sec:BayOptD}.  
Particularly, in the special case without any covariates, i.e., $\bm F=\bm 1_n$, 
\[
T_2(\bm x, \rho)=(1-\rho) \frac{(\bm x^\top\bm m)^2}{\sum^n_{i=1}m_i},
\]
where $\bm m=(m_1, \ldots, m_n)^\top$ with $m_i$ be the number of adjacent neighbors of the $i$-th user. 
Therefore, $\cor_{\bm x}(T_2(\bm x, \rho_0), T_2(\bm x, \rho))=1$ for any $\rho_0$ and $\rho$. 
It indicates that there is no loss to replace an unknown true $\rho$ with a given $\rho_0$ in this special case.


\section{Numerical Study}\label{sec:num}

The purpose of optimal design is to reduce the variance (or equivalently, improve the precision) of the estimated treatment effect $\hat\theta$ in \eqref{eq:lm}. 
Since the optimal value of the design criterion can not be obtained directly, computing the classical measure ``design efficiency'' is not feasible.
Alternatively, we evaluate the quality of design by computing the percentage of the improvement in precision compared to the expected precision of random balanced designs.
\begin{proposition}\label{prop:prec}
Consider a random balanced design $\bm x$. 
The marginal distribution of each $x_i$ is $\Pr(x_i=1)=\Pr(x_i=-1)=0.5$ and  $\sum^n_{i=1} x_i$ follows the balance condition, i.e., $-1\leq\sum^n_{i=1} x_i\leq 1$.
The expected precision of the random balanced design is
\begin{equation}\label{eq:expected}
\mathbb{E}_{\bm x} \left(\sigma^{-2}\bm x^\top\bm K\bm x\right)=\sigma^{-2}\mathrm{tr}(\bm K\bm C),
\end{equation}
where $\bm C$ is an $n\times n$ matrix with all of the diagonal entries equal to 1 and all of the off-diagonal entries equal to a fixed constant $c$. The  value of $c$ is $-(n-1)^{-1}$ if $n$ is even and it is $-n^{-1}$ if $n$ is odd.
Here $\mathrm{tr}(\cdot)$ denotes the trace of a matrix.
The expectation in \eqref{eq:expected} is taken with respect to the probability distribution of $\bm x$. 
\end{proposition}

The proof of the above proposition is given in the Supplement. For any given design $\bm x_0$, the percentage of the improvement in precision with respect to the expected precision of the random balanced design can be expressed by
\begin{equation}\label{eq:precIM}
\text{PIP}(\bm x_0)\:=\frac{\sigma^{-2}\bm x^\top_0 \bm K\bm x_0-\mathbb{E}_{\bm x} \left(\sigma^{-2}\bm x^\top\bm K\bm x\right)}{\sigma^{-2}\bm x^\top_0 \bm K\bm x_0}
=1-\frac{\tr(\bm K\bm C)}{\bm x^\top_0 \bm K\bm x_0}.
\end{equation}
For short, we denote this percentage of improvement in precision by $\text{PIP}(\bm x_0)$. 
According to \eqref{eq:K}, the calculation of matrix $\bm K$ involves the network correlation parameter $\rho$.
Since \eqref{eq:precIM} is used to evaluate the design $\bm x_0$, naturally, we should use the true value of the network correlation, denoted by $\rho_t$, to compute $\text{PIP}(\bm x_0)$.
In the following simulation study, $\rho_t$ is part of the simulation settings.

In Section \ref{sec:robust}, we evaluate the robustness of the proposed design approach to different choices of $\alpha$ and $\rho_0$.
In Section \ref{sec:value}, we evaluate the advantages of the optimal design with network connection under different scenarios. 
In both subsections, we generate synthetic datasets, where the edges of the network are independently generated from a Bernoulli distribution with a constant probability, which is called \emph{network density}. 
If there are isolated nodes in the generated network, we connect each of them with a randomly selected neighbor to remove isolation and ensure that $m_i\geq 1$ in \eqref{eq:delta}.
Each node is associated with a $p$-dimensional covariates whose entries are randomly generated from $\{-1, 1\}$ with equal probabilities. 
To stabilize the results, we generate 10 copies of datasets and report the results in boxplots.

\subsection{Robustness on the Choices of $\alpha$ and $\rho_0$}\label{sec:robust}

In this subsection, we consider two versions of the proposed hybrid design approach.
\begin{itemize}
\item[1.] Locally optimal design: the optimal design obtained by solving the optimization problem in \eqref{eq:opt_bayesian}.
We specify the mean of the prior distribution to be 0.5, i.e., $\rho_0=0.5$. 
\item[2.] True optimal design: the optimal design obtained by maximizing the objective in \eqref{eq:max} with the true network correlation value $\rho_t$.
\end{itemize}
We use the hybrid approach in \eqref{eq:formulation2} to obtain both the locally and true optimal designs.
The comparison between the locally optimal design and the true optimal design shows the gap of replacing the true design criterion $T(\bm x,\rho_t)$ by its practical surrogate $T(\bm x,\rho_0)$.
For both designs, we use Gurobi \citep{gurobi2015gurobi} to solve the optimization problem, and the run-time is limited to 500 seconds.

First, we evaluate the performance of the locally optimal design with $\rho_0=0.5$ with different choices of $\alpha$. 
In this case, we fix $p=10$ and the network density is 0.08. 
Each boxplot in Figure \ref{fg:smallalpha} shows the $\text{PIP}(\bm x)$ values of 10 datasets. 
We can detect a slightly bigger $\text{PIP}(\bm x)$ for smaller $\alpha$ values.
However, this trend diminishes when $\alpha=0.001$ and $\alpha=0.0001$. 
Therefore, we set $\alpha=0.001$ for all subsequent test cases and real case studies. 
As stated in Section \ref{sec:hybrid}, we recommend a sensitivity check in practice. 
Since it is computationally efficient to compute the optimal design and $\text{PIP}(\bm x)$ value, the experimenter can obtain the optimal design for a sequence of $\alpha$ values and choose the one when further decreasing $\alpha$ does not increase PIP. 

\begin{figure}
\centering
\includegraphics[scale=0.7]{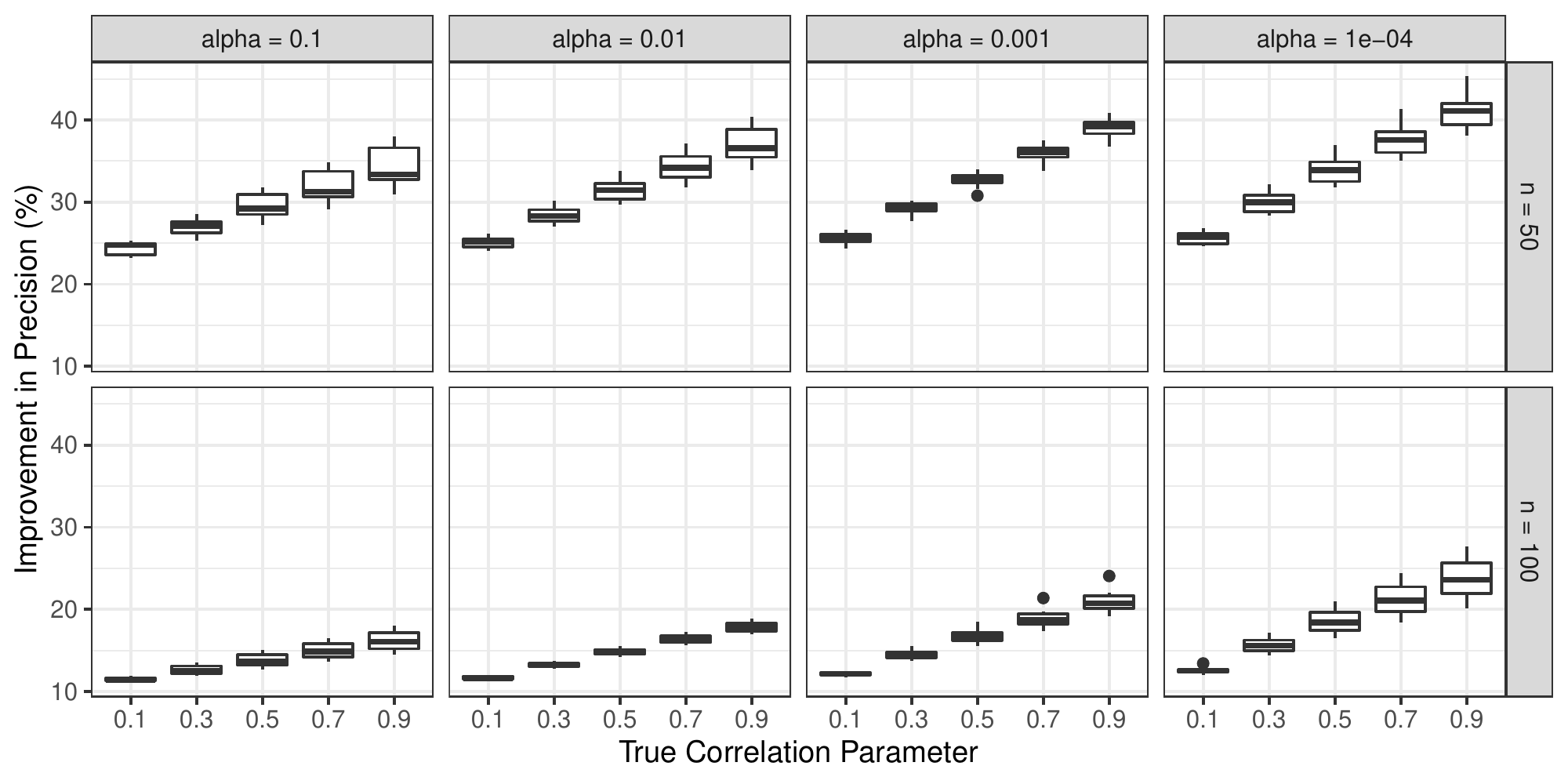}
\caption{$\text{PIP}(\bm x)$ of the locally optimal designs with $p=10$ and network density 0.08.}\label{fg:smallalpha}
\end{figure}

The second simulation is to evaluate the robustness of the locally optimal design to the choice of $\rho_0$. 
We fix $\alpha=0.001$ and the network density be 0.08. 
As discussed in Section \ref{sec:BayOptD}, it is expected that the difference between the locally optimal design and the true optimal design is small. 
Figure \ref{fg:smallrho} confirms this. 
It shows the boxplots of the differences between $\text{PIP}$ of the two designs for the same data. 
Each boxplot is based on 10 replications.
According to Figure \ref{fg:smallrho}, the differences are mostly under 3\%. 
Theoretically, the PIP based on $\rho_t$ should be strictly larger than the ones based on $\rho_0$ if the designs are solutions to the original optimization \eqref{eq:opt_bayesian}. 
But the solutions in Figure \ref{fg:smallrho} are based on the hybrid approach, so the PIP based on $\rho_t$ can be sometimes smaller than the PIP based on $\rho_0$. 
Since the locally optimal design with different $\rho_0$ performs similarly to the true optimal design with $\rho_t$ in terms of PIP, we use the locally optimal design with $\rho_0=0.5$ for the rest of this section.

\begin{figure}
\centering
\includegraphics[scale=0.7]{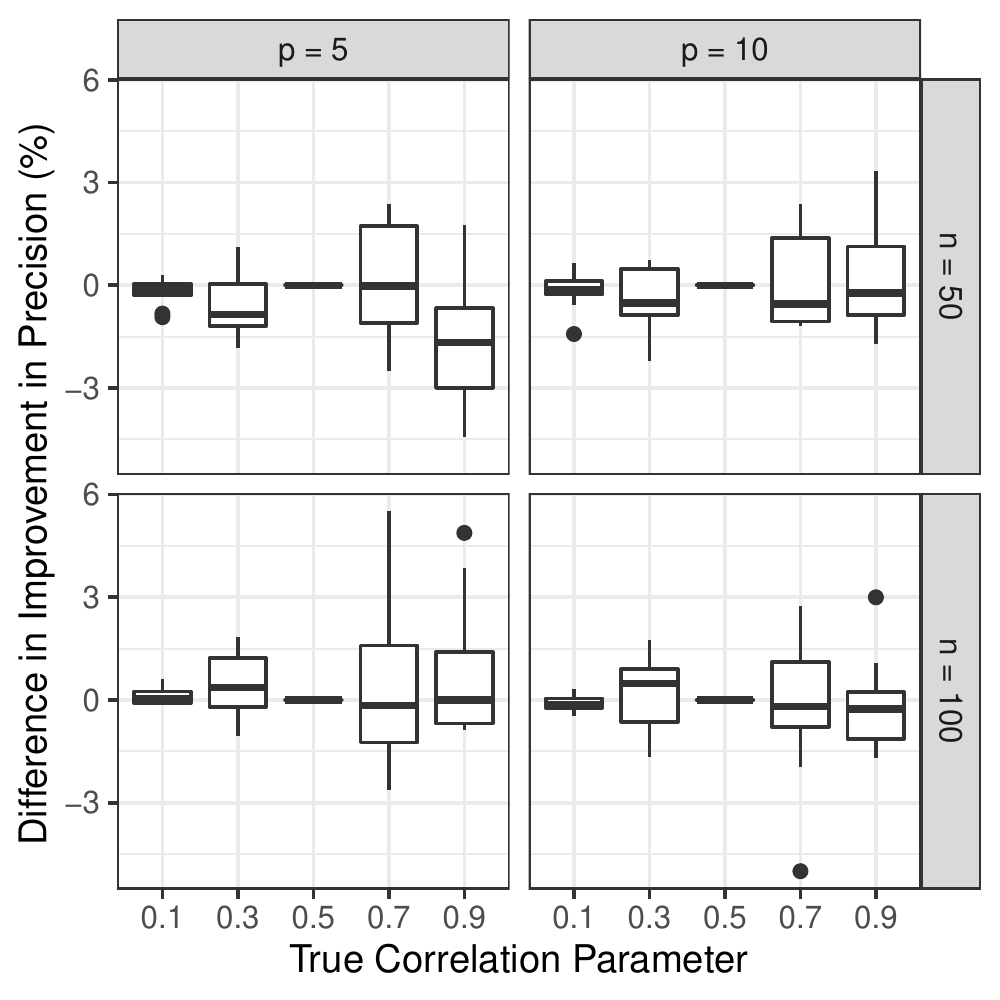}
\caption{The differences of PIP between the locally optimal design with $\rho_0=0.5$ and the true optimal design with $\rho_t$ for $n=50,100$ and $p=5,10$.}\label{fg:smallrho} 
\end{figure}

\subsection{The Advantage of Considering Network Connection}\label{sec:value}

In this subsection, we address the advantage of considering network connection in the design procedure.
First, we comment on the influence of the network on the performance of the proposed optimal design. 
Essentially, the proposed locally optimal design aims to maximize $T(\bm x,\rho_0)$, which is equal to $m-T_1-T_2$. 
Recall that $m$ is twice the total number of edges and it increases as the network expands in the number of nodes $n$ and/or network density. 
Therefore, for large and dense network, $m$ can dominate the objective function $T(\bm x, \rho_0)$. 
Since the PIP \eqref{eq:precIM} is calculated based on $T(\bm x,\rho_0)$, the advantage of the proposed optimal design would appear to be marginal. 
When $n$ is not large and the network connection is sparse, the advantage of the proposed optimal design would be more significant. 
In the following, the locally optimal design is compared with the optimal design in \eqref{eq:opt_no} that does not consider network connections. 
The latter is obtained using Gurobi \citep{gurobi2015gurobi} and the run-time is also limited to 500 seconds, the same as the locally optimal design.

In Figure \ref{fg:smallnetp}, we compare the locally optimal designs with network and the optimal design without network under different network densities.
For each synthetic dataset, we obtain the two different optimal designs (i.e., with and without network connection) and obtain their respective PIP values. 
The boxplots are PIP values for 10 synthetic datasets under the same $\rho_t$, $n$, and network density setting. 
The results indicate that by incorporating the network structure, the proposed locally optimal design significantly outperforms the optimal design without a network connection, and this advantage is more prominent when network density is lower, the network size $n$ is smaller, and the true value of correlation $\rho_t$ is larger.

\begin{figure}
\centering
\includegraphics[scale=0.7]{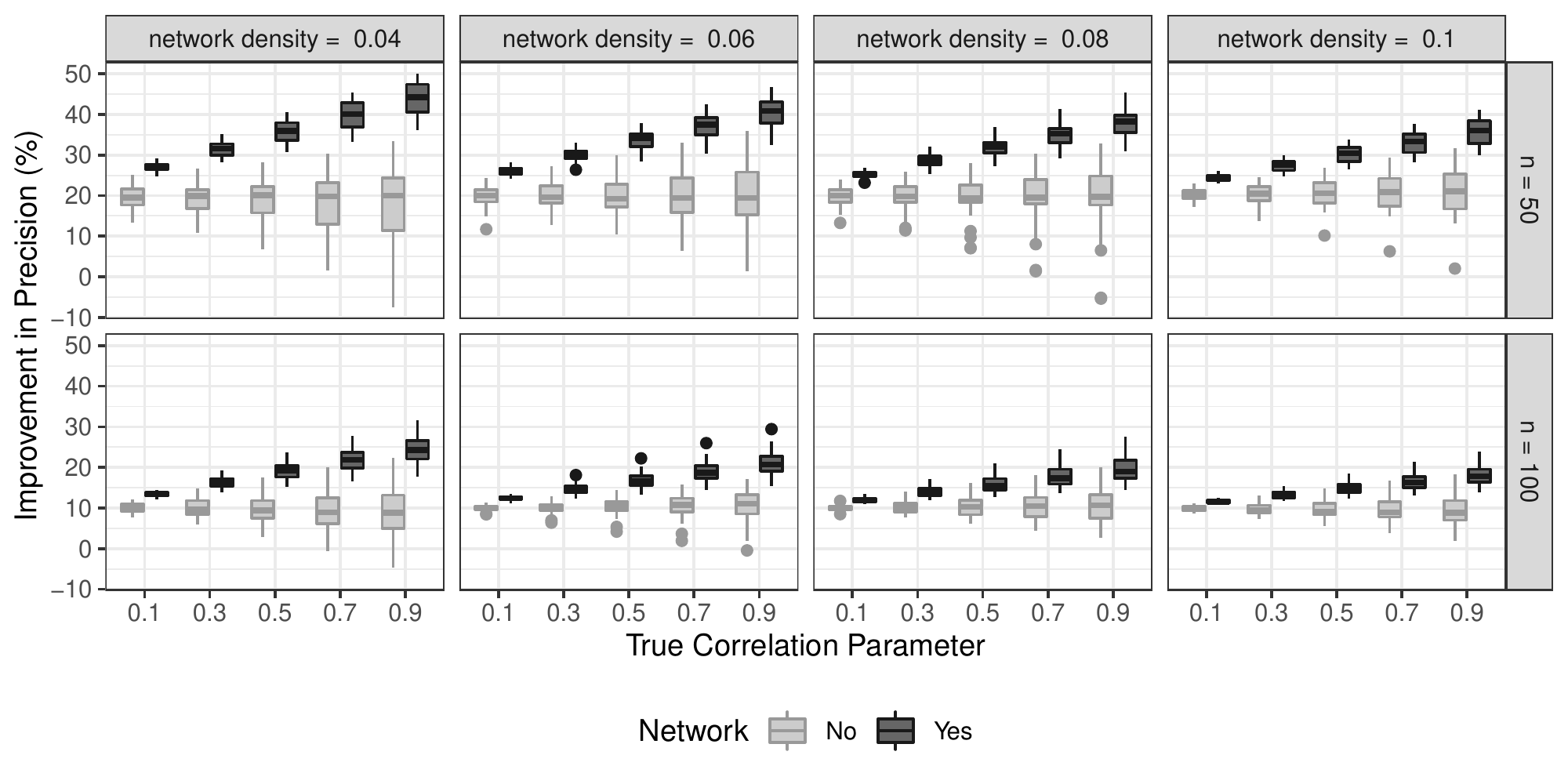}
\caption{The PIP values of optimal designs with and without network with $p=10$.}\label{fg:smallnetp}
\end{figure}

Similarly, in Figure \ref{fg:n}, we expand such comparison to more cases of $n=50, 100, 500, 1000$.
In addition to the PIP in \eqref{eq:precIM}, we also include the improvement in $T_1(\bm x, \rho)$ and $T_2(\bm x, \rho)$ with respect to the expected $T_1$ and $T_2$ of the random balanced designs. 
Although for the locally optimal design with network, PIP value drops from 40\% to 5\% as $n$ increases from 50 to 1000, the improvements in $T_1$ and $T_2$ do not decrease with the network size $n$.  
For instance, using the proposed optimal design criterion for the cases with $n=1000$, $m\approx 20,000$, and $\rho_t=0.1,\ldots, 0.9$, $T_1$ varies from -300 to -20 and $T_2$ varies from 0 to 200.
As discussed earlier, the main reason is that $m$ dominates the percentage of improvement in precision when $n$ is large and/or the network is dense.

\begin{figure}
\centering
\includegraphics[scale=0.7]{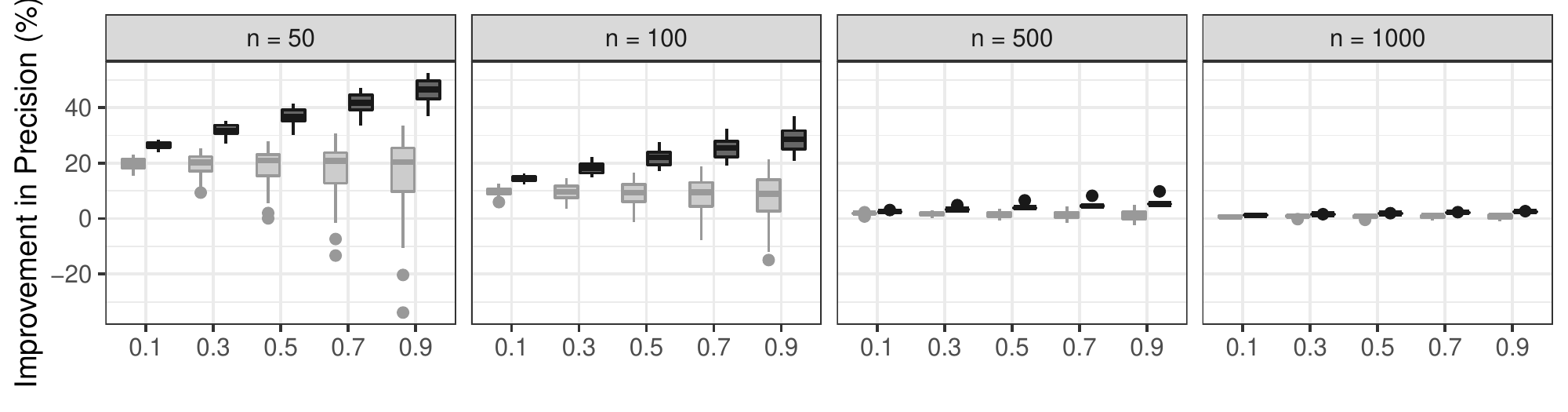}
\includegraphics[scale=0.7]{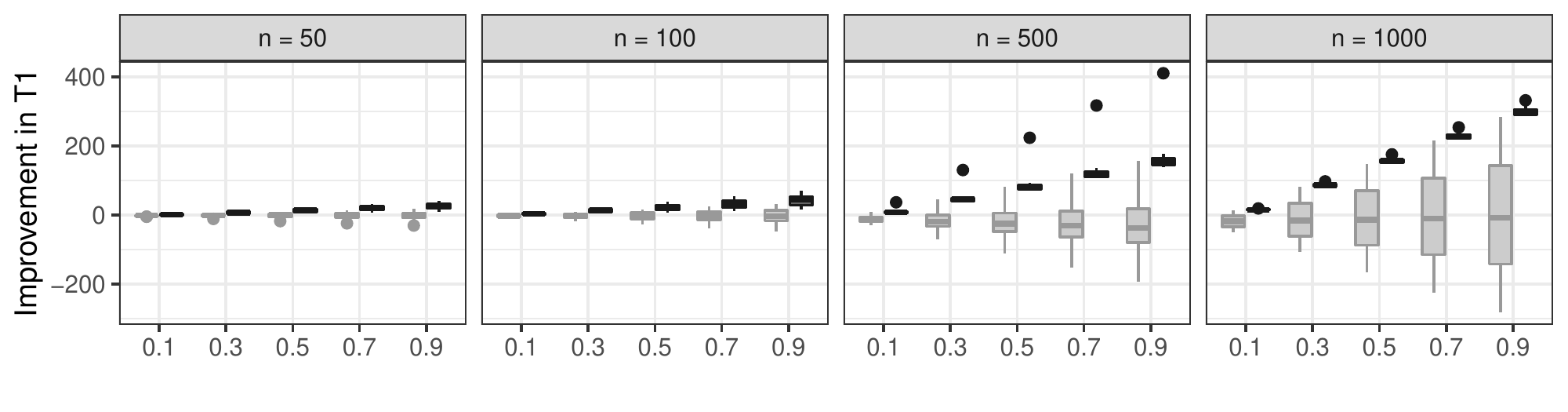}
\includegraphics[scale=0.7]{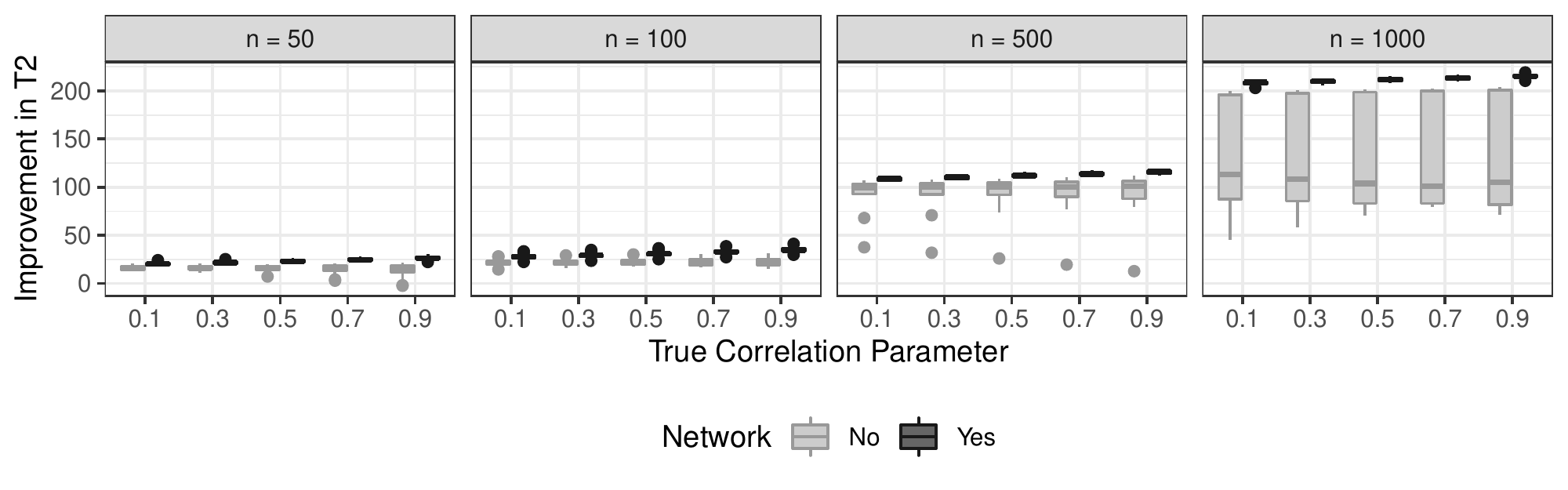}
\caption{The PIPs of the locally optimal design with network and the optimal design without network with $p=10$ and network density 0.02.}\label{fg:n}
\end{figure}

\section{Case Study}\label{sec:real}

The case study is based on a real dataset from \cite{data}, which is collected from the music streaming service (November 2017) with a total number of 47538 users from Hungary. 
The dataset contains the information of friendship networks of the users, as well as their covariates information, representing the users' preference (recorded by 1 or 0) to 84 distinct music genres.
To decide if the update of the music recommender algorithm improves the baseline algorithm, a controlled experiment can be conducted, much similar to the application context given in Section \ref{sec:intro}. 
The outcome of each user can be the total time of the user listening to the recommended music or a more direct metric commonly used by the company. 
The estimation of the treatment effect $\theta$ in \eqref{eq:lm} would reveal which one of the two versions of the recommender algorithm outperforms the other. 
Both the social network of users and their covariates are relevant in assessing different algorithms for music recommender systems. 
The non-interference assumption is proper for this case study since we assume the experiment is conducted without users' awareness. 

To evaluate the performance of the proposed design approaches, we repeatedly randomly sample sub-networks with 2000 and 3000 users from the complete data of 47538 users. 
Among those, around half of the users are isolated from other users (i.e., no network connections at all). 
This number is big due to the subset sampling of the original complete network. 
The CAR model does not work for isolated users, since $m_i$ has to be larger than zero. 
For simplicity, we remove those isolated users and the size of the remaining networks is approximately 1000 or 2000. 
In practice, all the isolated users can still be kept in the experiment and split into two groups via a covariate balancing measure. 
The densities of the resulting sub-networks range from 0.001 to 0.002.
Although the complete data contains 84 distinct genres as the covariates, many of them are linearly dependent. 
Also, because of subset sampling, many covariates of the subsets become constants. 
Thus, we keep the first $20$ covariates to remove the potential singularity issue. 
In the numerical study, we set $p$ from 5 to 20.

We first compute the PIP values given in \eqref{eq:precIM}. 
For the locally optimal design, we set $\alpha=0.001$ and $\rho_0=0.5$. 
The true correlation parameter $\rho_t$ is varied from 0.1 to 0.9. 
We include the optimal design without network connection in \eqref{eq:opt_no} for comparison. 
For each $n$ and $p$, ten subsets are randomly sampled from the complete data. 
The results are shown in Figure \ref{fg:real}. 
The results based on the real data are different from synthetic datasets in many aspects.
For instance, the distributions of covariates and networks are more complex. 
Particularly, we compute the proportions of $1$'s for each covariate, and it ranges in $[0.05, 0.85]$. 
The correlations between different covariates are in $[0.04, 0.97]$, so some of the covariates are highly correlated. 
Still, the results in Figure \ref{fg:real} show a similar pattern to the ones from synthetic networks, which indicates that the proposed approach is effective for real data sets as well, despite the more complicated network structure and covariates distributions.

\begin{figure}
\centering
\includegraphics[scale=0.7]{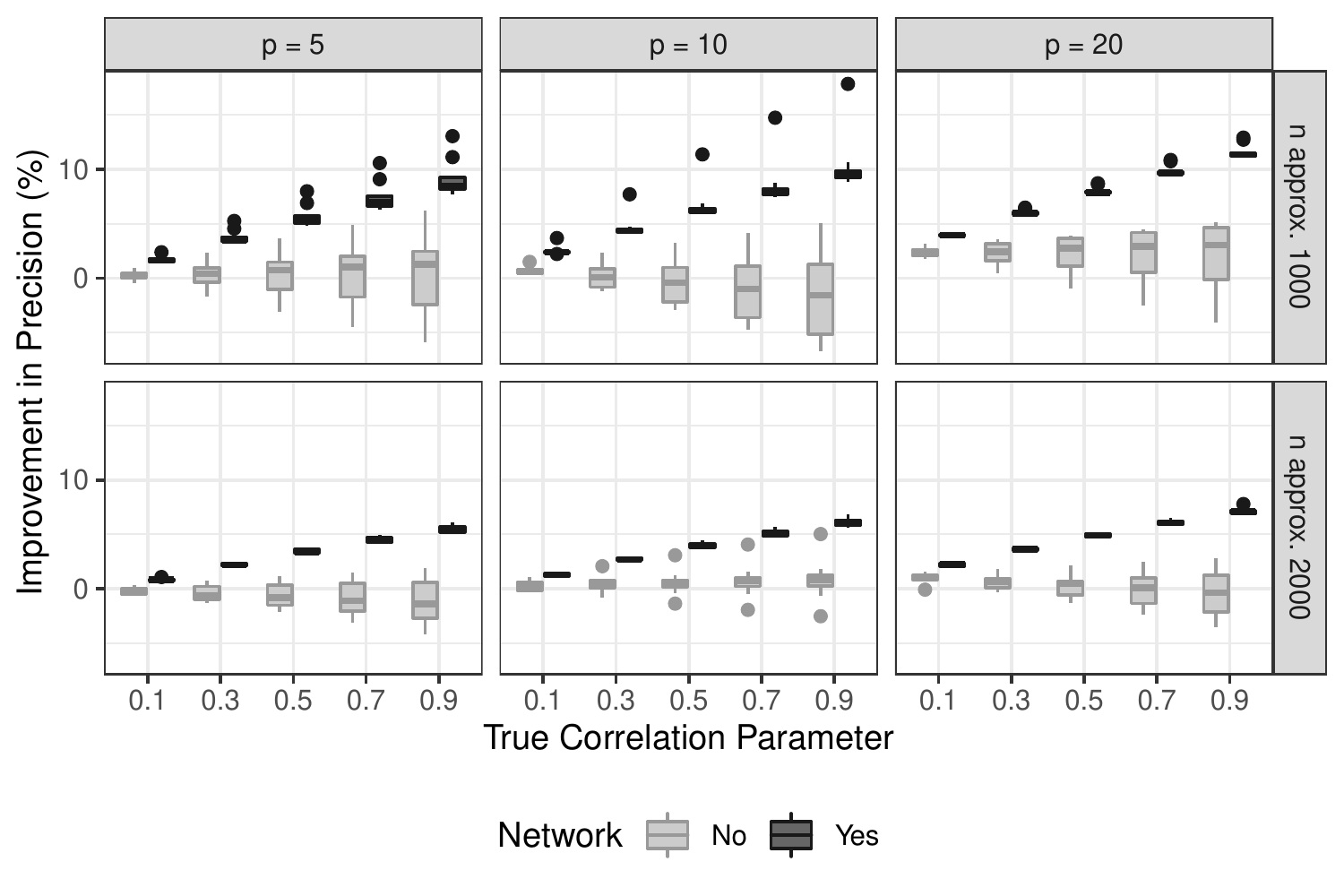}
\caption{Boxplots of percentages of PIP values of two kinds of optimal designs for case study.}\label{fg:real}
\end{figure}

Next, we create a pseudo experiment by simulating the outcome data and then compare the two methods empirically. 
In reality, the network correlation coefficient for different users may not be the same. 
Therefore, to generate the outcome data, we set the covariance matrix of the CAR model to be $\sigma^2(\bm D-\bm P\bm W\bm P)^{-1}$, where $\bm P$ is a diagonal matrix with entries $\sqrt{\rho_1}, \ldots, \sqrt{\rho_n}$.  
The heterogenous correlation coefficients $\rho_1, \ldots, \rho_n$ are sampled from the uniform distribution $U(0, 1)$.
We simulate the outcomes from the CAR model in \eqref{eq:car1} under this covariance structure, and then fit a CAR model with a single unknown correlation coefficient and estimate the treatment effect $\theta$ in \eqref{eq:lm}. 
We sample sub-networks with approximately 1000 users and take the first $p=5$ or 20 covariates. 
The true treatment effect $\theta$ and the variance $\sigma$ are specified to be 1. 
In addition to the locally optimal design and the optimal design without the network, we also generate 10 random balanced designs. 
For each design, we use the simulated outcome to obtain the estimate $\hat\theta$ based on the CAR model. 
Repeating this procedure 100 times, we compute the mean squared errors (MSEs) for each design approach. 
During each of the 100 times of simulation for each sampled sub-network and $p$ covariates, we obtain 12 MSEs values (2 optimal designs and 10 random designs) and then compute the empirical percentiles of the MSEs of two optimal designs respectively from the MSEs of the 10 random designs. 
If the empirical percentile of the MSE from an optimal design is smaller than 0.5, it means that the MSE of the optimal design is superior to more than 50\% of the random designs in terms of reducing the MSE.
Notice that the resulting empirical percentiles vary from different sub-datasets in each simulation. 
For each $p$, we generate 25 random sub-datasets. 
The empirical percentiles of MSEs of the two optimal designs are shown in the boxplots in Figure \ref{fg:realmse} for all the 25 sub-networks and two $p$ values. 
The results show that the optimal design without the network does not outperform the random balanced designs. 
For the proposed locally optimal design with the network, the empirical percentiles are mostly below 0.5, which strongly indicates its advantage over the other two alternatives in terms of reducing MSE.


\begin{figure}[ht]
\centering
\includegraphics[scale=0.7]{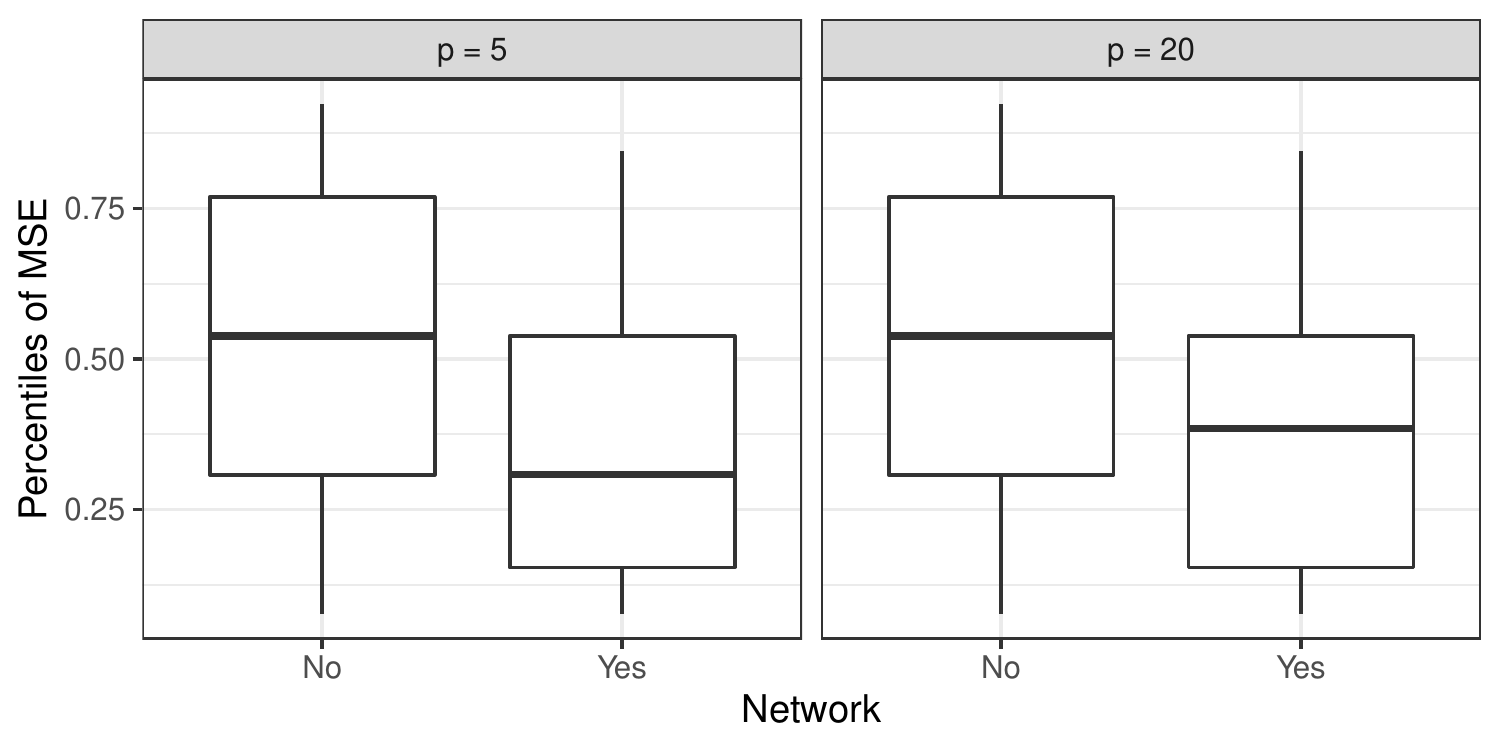}
\caption{The percentiles of MSEs of the two optimal design approaches (i.e., with and without network) based on the MSEs of 10 completely randomized designs.}\label{fg:realmse}
\end{figure}

\section{Conclusion}\label{sec:con}

In this paper, we propose a model-based optimal design approach to include both covariates and network dependence for the experiments of A/B tests.
A linear additive model is used to include the covariates information and the CAR model is used to model the network correlation between test subjects.
A hybrid approach is proposed to solve the optimization problem and construct the locally optimal design.
Both simulation and real data are used to compare the performances of the proposed locally optimal design with the exact optimal design and other alternative approaches.
The proposed design performances reasonably well compared with other approaches in terms of variance reduction to random designs.
The proposed optimal design relies on the CAR linear additive model including both covariates information and network correlation.
Similarly to all optimal design approaches, the validity of the model assumption is crucial.
Although we have shown the proposed design has some degree of robustness to the choice of the correlation parameter, if the experimenter thinks the CAR-based additive model assumption does not apply to the potential data to be collected, we recommend the rerandomization approaches proposed by \cite{morgan2012rerandomization} and \cite{morgan2015rerandomization} or completely randomized design if the sample size is sufficiently large.

We would like to point out a few directions for future research. 
First, this work is limited to the CAR model assumption and the network is much simpler than real social networks. 
But the proposed design approach can be applied to more sophisticated parametric models.
For example, the network can become directional and weighted, which can be specified by the adjacency matrix.
Other than the CAR model, the Spatial Auto-Regressive (SAR) model can be used.
The network correlation parameter $\rho$ can be different for different subjects as in the pseudo experiment we show in Section \ref{sec:real}. 
The CAR model can also be adjusted by changing the variance in \eqref{eq:delta} to handle the isolated nodes, which can exist occasionally in real social networks. 
Second, for the extremely large networks, there may be a time or economic cost to involve as many test subjects as possible.
In this case, the optimal design proposed here can be extended to the optimization problem of simultaneous selection of test subjects and treatment assignment.
Some computational efficient approximation algorithms need to be adapted to solve this problem for large networks.
Third, the proposed design relies on the observed covariates, which might be inaccurate depending on the data source.
To make the design robust to inaccurate covariates, we may incorporate the uncertainty of those covariates, and develop a hierarchy model that can characterize the uncertainty.
How to design treatment allocation under this situation would be an interesting topic. 
Moreover, it is also important to investigate designs when the number of treatment settings is more than two, particularly when the experiment involves multiple factors. 

\section*{Acknowledgments}

The authors thank the Editor, the Associate Editor, and two reviewers for their valuable comments through the reviewing process. 

\section*{Funding}

This research is supported by a U.S. National Science Foundation grant DMS-1916467. 

\section*{Supplementary materials}

The supplementary materials include proofs, derivations, and extra examples. They also include the codes for all the examples. 

\bibliography{Ref.bib}

\begin{thebibliography}{35}
\newcommand{\enquote}[1]{``#1''}
\expandafter\ifx\csname natexlab\endcsname\relax\def\natexlab#1{#1}\fi

\bibitem[{Atkinson and Bailey(2001)}]{atkinson2001one}
Atkinson, A.~C. and Bailey, R. (2001), \enquote{One hundred years of the design
  of experiments on and off the pages of Biometrika,} \textit{Biometrika}, 88,
  53--97.

\bibitem[{Atkinson and Donev(1992)}]{atkinson1992optimum}
Atkinson, A.~C. and Donev, A.~N. (1992), \textit{Optimum experimental designs},
  Oxford Science Publications, London.

\bibitem[{Banerjee et~al.(2014)Banerjee, Carlin, and
  Gelfand}]{banerjee2014hierarchical}
Banerjee, S., Carlin, B.~P., and Gelfand, A.~E. (2014), \textit{Hierarchical
  Modeling and Analysis for Spatial Data}, New York: Chapman and Hall/CRC, 2nd
  ed.

\bibitem[{Basse and Airoldi(2018{\natexlab{a}})}]{basse2018limitations}
Basse, G.~W. and Airoldi, E.~M. (2018{\natexlab{a}}), \enquote{Limitations of
  design-based causal inference and A/B testing under arbitrary and network
  interference,} \textit{Sociological Methodology}, 48, 136--151.

\bibitem[{Basse and Airoldi(2018{\natexlab{b}})}]{basse2018model}
--- (2018{\natexlab{b}}), \enquote{Model-assisted design of experiments in the
  presence of network-correlated outcomes,} \textit{Biometrika}, 105, 849--858.

\bibitem[{Belotti et~al.(2013)Belotti, Kirches, Leyffer, Linderoth, Luedtke,
  and Mahajan}]{belotti2013mixed}
Belotti, P., Kirches, C., Leyffer, S., Linderoth, J., Luedtke, J., and Mahajan,
  A. (2013), \enquote{Mixed-integer nonlinear optimization,} \textit{Acta
  Numerica}, 22, 1--131.

\bibitem[{Bertsimas et~al.(2015)Bertsimas, Johnson, and
  Kallus}]{bertsimas2015power}
Bertsimas, D., Johnson, M., and Kallus, N. (2015), \enquote{The power of
  optimization over randomization in designing experiments involving small
  samples,} \textit{Operations Research}, 63, 868--876.

\bibitem[{Besag(1974)}]{besag1974spatial}
Besag, J. (1974), \enquote{Spatial interaction and the statistical analysis of
  lattice systems,} \textit{Journal of the Royal Statistical Society: Series B
  (Methodological)}, 36, 192--225.

\bibitem[{Bhat et~al.(2020)Bhat, Farias, Moallemi, and Sinha}]{bhat2020near}
Bhat, N., Farias, V.~F., Moallemi, C.~C., and Sinha, D. (2020),
  \enquote{Near-Optimal AB Testing,} \textit{Management Science}.

\bibitem[{Chaloner and Verdinelli(1995)}]{chaloner1995bayesian}
Chaloner, K. and Verdinelli, I. (1995), \enquote{Bayesian experimental design:
  A review,} \textit{Statistical Science}, 273--304.

\bibitem[{Cressie(1993)}]{cressie1993spatial}
Cressie, N. A.~C. (1993), \textit{Statistics for spatial data}, New York:
  Wiley, revised edition ed.

\bibitem[{Drovandi and Tran(2018)}]{drovandi2018improving}
Drovandi, C.~C. and Tran, M.-N. (2018), \enquote{{Improving the Efficiency of
  Fully Bayesian Optimal Design of Experiments Using Randomised Quasi-Monte
  Carlo},} \textit{Bayesian Analysis}, 13, 139 -- 162.

\bibitem[{Eckles et~al.(2016)Eckles, Karrer, and Ugander}]{eckles2015design}
Eckles, D., Karrer, B., and Ugander, J. (2016), \enquote{Design and Analysis of
  Experiments in Networks: Reducing Bias from Interference,} \textit{Journal of
  Causal Inference}, 5, 20150021.

\bibitem[{Gui et~al.(2015)Gui, Xu, Bhasin, and Han}]{gui2015network}
Gui, H., Xu, Y., Bhasin, A., and Han, J. (2015), \enquote{Network A/B Testing:
  From Sampling to Estimation,} in \textit{Proceedings of the 24th
  International Conference on World Wide Web}, pp. 399--409.

\bibitem[{Gurobi~Optimization(2015)}]{gurobi2015gurobi}
Gurobi~Optimization, I. (2015), \enquote{Gurobi optimizer reference manual,}
  \textit{URL http://www. gurobi. com}.

\bibitem[{Imbens and Rubin(2015)}]{imbens2015causal}
Imbens, G.~W. and Rubin, D.~B. (2015), \textit{Causal Inference for Statistics,
  Social, and Biomedical Sciences: An Introduction}, New York: Cambridge
  University Press.

\bibitem[{Kallus(2018)}]{kallus2018optimal}
Kallus, N. (2018), \enquote{Optimal a priori balance in the design of
  controlled experiments,} \textit{Journal of the Royal Statistical Society:
  Series B (Statistical Methodology)}, 80, 85--112.

\bibitem[{Kiefer(1961)}]{kiefer1961optimum}
Kiefer, J. (1961), \enquote{Optimum designs in regression problems, II,}
  \textit{The Annals of Mathematical Statistics}, 298--325.

\bibitem[{Kohavi et~al.(2020)Kohavi, Tang, and Xu}]{kohavi2020trustworthy}
Kohavi, R., Tang, D., and Xu, Y. (2020), \textit{Trustworthy online controlled
  experiments: A practical guide to a/b testing}, Cambridge University Press.

\bibitem[{Koutra(2017)}]{koutra2017designing}
Koutra, V. (2017), \enquote{Designing experiments on networks,} Ph.D. thesis,
  University of Southampton.

\bibitem[{Li et~al.(2021)Li, Kang, and Huang}]{li2021covariate}
Li, Y., Kang, L., and Huang, X. (2021), \enquote{Covariate balancing based on
  kernel density estimates for controlled experiments,} \textit{Statistical
  Theory and Related Fields}, 5, 102--113.

\bibitem[{Martin(1986)}]{martin1986design}
Martin, R. (1986), \enquote{On the design of experiments under spatial
  correlation,} \textit{Biometrika}, 73, 247--277.

\bibitem[{Morgan and Rubin(2012)}]{morgan2012rerandomization}
Morgan, K.~L. and Rubin, D.~B. (2012), \enquote{Rerandomization to improve
  covariate balance in experiments,} \textit{The Annals of Statistics}, 40,
  1263--1282.

\bibitem[{Morgan and Rubin(2015)}]{morgan2015rerandomization}
--- (2015), \enquote{Rerandomization to balance tiers of covariates,}
  \textit{Journal of the American Statistical Association}, 110, 1412--1421.

\bibitem[{Nandy et~al.(2020)Nandy, Basu, Chatterjee, and Tu}]{nandy2020b}
Nandy, P., Basu, K., Chatterjee, S., and Tu, Y. (2020), \enquote{A/B testing in
  dense large-scale networks: design and inference,} \textit{Advances in Neural
  Information Processing Systems}, 33.

\bibitem[{Parker et~al.(2017)Parker, Gilmour, and
  Schormans}]{parker2017optimal}
Parker, B.~M., Gilmour, S.~G., and Schormans, J. (2017), \enquote{Optimal
  design of experiments on connected units with application to social
  networks,} \textit{Journal of the Royal Statistical Society: Series C
  (Applied Statistics)}, 3, 455--480.

\bibitem[{Phan and Airoldi(2015)}]{phan2015natural}
Phan, T.~Q. and Airoldi, E.~M. (2015), \enquote{A natural experiment of social
  network formation and dynamics,} \textit{Proceedings of the National Academy
  of Sciences}, 112, 6595--6600.

\bibitem[{Pokhilko et~al.(2019)Pokhilko, Zhang, Kang, and
  Darcy}]{pokhilko2019d}
Pokhilko, V., Zhang, Q., Kang, L., and Darcy, P.~M. (2019), \enquote{D-Optimal
  Design for Network A/B Testing,} \textit{Journal of Statistical Theory and
  Practice}, 13, 61.

\bibitem[{Rozemberczki et~al.(2018)Rozemberczki, Davies, Sarkar, and
  Sutton}]{data}
Rozemberczki, B., Davies, R., Sarkar, R., and Sutton, C. (2018),
  \enquote{GEMSEC: Graph Embedding with Self Clustering,} .

\bibitem[{Rubin(1974)}]{rubin1974estimating}
Rubin, D.~B. (1974), \enquote{Estimating causal effects of treatments in
  randomized and nonrandomized studies,} \textit{Journal of educational
  Psychology}, 66, 688.

\bibitem[{Rubin(2005)}]{rubin2005causal}
--- (2005), \enquote{Causal inference using potential outcomes: Design,
  modeling, decisions,} \textit{Journal of the American Statistical
  Association}, 100, 322--331.

\bibitem[{Rue and Held(2005)}]{rue2005gaussian}
Rue, H. and Held, L. (2005), \textit{Gaussian Markov Random Fields: Theory and
  Applications}, New York: Chapman and Hall/CRC.

\bibitem[{Ryan et~al.(2016)Ryan, Drovandi, McGree, and
  Pettitt}]{ryan2016review}
Ryan, E.~G., Drovandi, C.~C., McGree, J.~M., and Pettitt, A.~N. (2016),
  \enquote{A Review of Modern Computational Algorithms for Bayesian Optimal
  Design,} \textit{International Statistical Review}, 84, 128--154.

\bibitem[{Ryan et~al.(2014)Ryan, Drovandi, Thompson, and
  Pettitt}]{ryan2014towards}
Ryan, E.~G., Drovandi, C.~C., Thompson, M.~H., and Pettitt, A.~N. (2014),
  \enquote{Towards Bayesian experimental design for nonlinear models that
  require a large number of sampling times,} \textit{Computational Statistics
  \& Data Analysis}, 70, 45--60.

\bibitem[{Ver~Hoef et~al.(2018)Ver~Hoef, Hanks, and
  Hooten}]{ver2018relationship}
Ver~Hoef, J.~M., Hanks, E.~M., and Hooten, M.~B. (2018), \enquote{On the
  relationship between conditional (CAR) and simultaneous (SAR) autoregressive
  models,} \textit{Spatial statistics}, 25, 68--85.

\end{thebibliography}

\newpage

\begin{center}
{\Large\bf Supplement: Proofs, Derivations and Extra Example}\\
Qiong Zhang$^{1}$, Lulu Kang$^{2}$\\
$^{1}$School of Mathematical and Statistical Sciences, Clemson University\\
$^{2}$Department of Applied Mathematics, Illinois Institute of Technology
\end{center}

\setcounter{figure}{0}
\setcounter{table}{0}
\setcounter{lem}{0}
\setcounter{theorem}{0}
\setcounter{proposition}{0}

\makeatletter 
\renewcommand{\thefigure}{S\@arabic\c@figure}
\renewcommand{\thetable}{S\@arabic\c@table}
\renewcommand{\thelem}{S\@arabic\c@lem}
\renewcommand{\theproposition}{S\@arabic\c@proposition}
\renewcommand{\thetheorem}{S\@arabic\c@theorem}
\makeatother

\section*{S1. Proof of Theorem \ref{thm:concave}}
\begin{proof}
Because $T(\bm x, \rho)=m-T_1(\bm x, \rho)-T_2(\bm x, \rho)$,
we investigate the derivatives $T_1$ and $T_2$ with respect to $\rho$. For $T_1$,
\begin{align*}
T_1(\bm x, \rho)&=\rho \bm x^\top \bm W\bm x,\\
\frac{\partial T_1(\bm x, \rho)}{\partial \rho} &= \bm x^\top \bm W\bm x, \quad \frac{\partial^2 T_1(\bm x,\rho)}{\partial \rho^2} = 0.
\end{align*}

To derive the derivatives for $T_2$, we first introduce some notation to shorten the formulas.
Let $\bm A:=\bm F^\top (\bm D-\rho \bm W)F$, $\bm A_1:=\frac{\partial \bm A^{-1}}{\partial \rho}$, and $\bm A_2:=\frac{\partial^2 \bm A^{-1}}{\partial \rho^2}$.
Following the calculus of matrix,
\begin{align*}
\bm A_1&=-\bm A^{-1}\frac{\partial \bm A}{\partial \rho}\bm A^{-1}=-\bm A^{-1}\frac{\partial \bm F^\top \bm D\bm F-\rho \bm F^\top\bm W\bm F}{\partial \rho}\bm A^{-1}=\bm A^{-1}\bm F^\top \bm W\bm F\bm A^{-1}, \\
\bm A_2&=\frac{\partial \bm A_1}{\partial \rho}=\frac{\partial \bm A^{-1}}{\partial \rho}\bm F^\top \bm W\bm F \bm A^{-1}+\bm A^{-1}\bm F^\top \bm W\bm F\frac{\partial \bm A^{-1}}{\partial \rho}=2\bm A^{-1}\bm F^\top \bm W\bm F\bm A^{-1}\bm F^\top\bm W\bm F\bm A^{-1}.
\end{align*}
Using the new notation,
\begin{align*}
T_2(\bm x, \rho)&=\bm x^\top (\bm D-\rho \bm W)\bm F \bm A^{-1} \bm F^\top (\bm D-\rho \bm W)\bm x\\
&=\underbrace{\bm x^\top \bm D\bm F\bm A^{-1}\bm F^\top\bm D\bm x}_{\text{Term 1}}-2\underbrace{\rho \bm x^\top \bm W\bm F\bm A^{-1}\bm F^\top\bm D\bm x}_{\text{Term 2}}+\underbrace{\rho^2\bm x^\top \bm W\bm F\bm A^{-1}\bm F^\top \bm W\bm x}_{\text{Term 3}}.
\end{align*}
The first order derivative of the three terms with respect to $\rho$ are
\begin{align*}
\frac{\partial \text{Term 1}}{\partial \rho}&=\bm x^\top \bm D\bm F\frac{\partial \bm A^{-1}}{\partial \rho}\bm F^\top \bm D\bm x=\bm x^\top\bm D\bm F\bm A_1F^\top \bm D\bm x\\
\frac{\partial \text{Term 2}}{\partial \rho}&=\bm x^\top \bm W\bm F\bm A^{-1}\bm F^\top\bm D\bm x+\rho \bm x^\top\bm W\bm F\frac{\partial \bm A^{-1}}{\partial \rho}\bm F^\top \bm D\bm x\\
&=\bm x^\top \bm W\bm F\bm A^{-1}\bm F^\top\bm D\bm x+\rho \bm x^\top\bm W\bm F\bm A_1\bm F^\top\bm D\bm x\\
\frac{\partial \text{Term 3}}{\partial \rho}&=2\rho \bm x^\top \bm W\bm F\bm A^{-1}\bm F^\top\bm W\bm x+\rho^2\bm x^\top\bm W\bm F\frac{\partial \bm A^{-1}}{\partial \rho}\bm F^\top \bm W\bm x\\
&=2\rho \bm x^\top \bm W\bm F\bm A^{-1}\bm F^\top\bm W\bm x+\rho^2\bm x^\top\bm W\bm F\bm A_1\bm F^\top\bm W\bm x.
\end{align*}
To combine the three derivatives,
\[
\frac{\partial T_2(\bm x, \rho)}{\partial \rho}=\bm x^\top (\bm D-\rho\bm W)\bm F\bm A_1\bm F^\top (\bm D-\rho\bm W)\bm x-2\bm x^\top\bm W\bm F\bm A^{-1}\bm F^\top (\bm D-\rho\bm W)\bm x
\]
The derivative of $T(\bm x, \rho)$ is,
\begin{align*}
\frac{\partial T(\bm x,\rho)}{\partial \rho} &= -\bm x^\top \bm W\bm x-\bm x^\top (\bm D-\rho \bm W)\bm F\bm A_1\bm F^\top (\bm D-\rho \bm W)\bm x+2\bm x^\top \bm W\bm F\bm A^{-1}\bm F^\top (\bm D-\rho W)\bm x\\
&=-\bm x^\top \bm W\left[\bm I_n-\bm F\bm A^{-1}\bm F^\top (\bm D-\rho \bm W)\right]\bm x\\
&-\bm x^\top \left[(\bm D-\rho \bm W)\bm F\bm A^{-1}\bm F^\top -\bm I_n\right]\bm W\bm F\bm A^{-1}\bm F^\top (\bm D-\rho \bm W)\bm x\\
&=-\bm x^\top \left[\bm I_n-\bm F\bm A^{-1}\bm F^\top(\bm D-\rho \bm W)\right]^\top \bm W\left[\bm I_n-\bm F\bm A^{-1}\bm F^\top (\bm D-\rho \bm W)\right]\bm x.
\end{align*}
It is interesting to notice that $\bm s:=\left[\bm I_n-\bm F\bm A^{-1}\bm F^\top (\bm D-\rho \bm W)\right]\bm x$ can be considered as the residuals of regression model $\bm x=\bm F\bm \beta+\bm v$, where $\bm v$ is the vector with mean equal to $\bf 0$ and covariance matrix $\bm D-\rho \bm W$.
By the definition of the adjacency matrix,
\[
\frac{\partial T(\bm x, \rho)}{\partial \rho}=-\sum_{w_{i,j}=1}s_is_j.
\]
Thus, the sign of $\partial T(\bm x,\rho)/\partial \rho$ is uncertain and is possible to be either positive or negative.

Next, we compute the second order derivative of $T_2(\bm x,\rho)$ with respect to $\rho$.
\begin{align*}
\frac{\partial^2 \text{Term 1}}{\partial \rho^2} &= \bm x^\top \bm D\bm F \frac{\partial \bm A_1}{\partial \rho}\bm F^\top \bm D \bm x=\bm x^\top \bm D\bm F\bm A_2\bm F^\top \bm D \bm x,\\
\frac{\partial^2 \text{Term 2}}{\partial \rho^2} &= \bm x^\top \bm W \bm F \frac{\partial \bm A^{-1}}{\partial \rho}\bm F^\top \bm D\bm x+\bm x^\top \bm W\bm F\bm A_1F^\top \bm D\bm x+\rho \bm x^\top\bm W\bm F\frac{\partial \bm A_1}{\partial \rho}\bm F^\top \bm D\bm x\\
&=2\bm x^\top \bm W \bm F \bm A_1\bm F^\top \bm D\bm x+\rho \bm x^\top \bm W\bm F\bm A_2\bm F^\top \bm D\bm x,\\
\frac{\partial^2 \text{Term 3}}{\partial \rho^2}&=2\bm x^\top \bm W\bm F\bm A^{-1}\bm F^\top \bm W\bm x+2\rho\bm x^\top\bm W\bm F\frac{\partial \bm A^{-1}}{\partial \rho}\bm F^\top\bm W\bm x+2\rho \bm x^\top \bm W\bm F\bm A_1\bm F^\top \bm W\bm x\\
&+\rho^2\bm x^\top \bm W\bm F\frac{\partial \bm A_1}{\partial \rho}\bm F^\top \bm W\bm x\\
&=2\bm x^\top \bm W\bm F\bm A^{-1}\bm F^\top\bm W\bm x+4\rho \bm x^\top\bm W\bm F\bm A_1\bm F^\top\bm W\bm x+\rho^2\bm x^\top\bm W\bm F\bm A_2\bm F^\top\bm W\bm x.
\end{align*}
Let $\bm C:=\bm F\bm A^{-1}\bm F^\top\bm W\bm F$. Then
\begin{align*}
\frac{\partial^2 T_2(\bm x,\rho)}{\partial \rho^2}&=2\bm x^\top (\bm D-\rho \bm W)\bm C\bm A^{-1}\bm C^\top (\bm D-\rho\bm W)\bm x-4\bm x^\top\bm W\bm F\bm A^{-1}\bm C^\top (\bm D-\rho \bm W)\bm x\\
&+2\bm x^\top \bm W\bm F\bm A^{-1}\bm F^\top\bm W\bm x\\
&=2\bm x^\top \left[(\bm D-\rho\bm W)\bm C-\bm W\bm F\right]\bm A^{-1}\left[\bm C^\top (\bm D-\rho \bm W)-\bm F^\top\bm W\right]\bm x.
\end{align*}
For any $\rho\in (0, 1)$, it is apparent that $\bm D-\rho \bm W$ is the Laplacian matrix of the weighted undirected graph with the constant weight $\rho$ for each edge, and it is also clear that $\bm D-\rho\bm W$ is a positive definite matrix.
We assume $\bm F$ is a full rank matrix so that the regression model is valid.
So $\bm A$ and $\bm A^{-1}$ are both positive definite.
Thus, $\frac{\partial^2 T_2(\bm x,\rho)}{\partial \rho^2}\geq 0$ and $\frac{\partial^2 T(\bm x,\rho)}{\partial \rho^2}\leq 0$ for any $\rho\in (0,1)$.
The design criterion $T(\bm x, \rho)$, which is to be maximized, is concave.
\end{proof}

\section*{S2. The Gap between $\mathbb{E}[T(\bm x, \rho)]$ and $T(\bm x, \rho_0)$}

We randomly generate a network of size $n=50$.
For each pair of nodes, an edge will connect the two with a probability of $1/4$ and the existence of the edge is independent of any other random variables.
The covariate $z_i$ is generated from a one-dimensional normal distribution $N(0, 10^2)$ and $z_i$'s are independent of each other and the network structure.
The prior distribution of $\rho$ is uniform distribution in $[0,1]$ and $\rho_0=\mathbb{E}(\rho)=1/2$.
We randomly generate 400 completely randomized designs $\bm x_l$ for $l=1,\ldots, 400$ and calculate $T(\bm x_l, \rho_0)$, whose histogram is plotted in the left panel of Figure \ref{fg:gap}.
For any given design $\bm x_l$, we randomly samples $\rho_i$ for $i=1,\ldots, 200$ and calculate $T(\bm x_l,\rho_i)$.
The mean $\mathbb{E}[T(\bm x_l,\rho)]$ is approximated by the sample mean of $T(\bm x_l,\rho_i)$'s.
The histogram of the gap $T(\bm x_l,\rho_0)-\mathbb{E}[T(\bm x_l,\rho)]$ for all the random designs is plotted in the right panel of Figure \ref{fg:gap}.
Based on the two histograms, the gap $T(\bm x,\rho_0)-\mathbb{E}[T(\bm x,\rho)]$ is relatively small compared to the range of $T(\bm x,\rho_0)$.
Thus, it is reasonable to use the surrogate local design criterion $T(\bm x_l,\rho_0)$ to replace $\mathbb{E}[T(\bm x, \rho)]$ for this simple example.

\begin{figure}[ht]
\centering
\includegraphics[width=\textwidth]{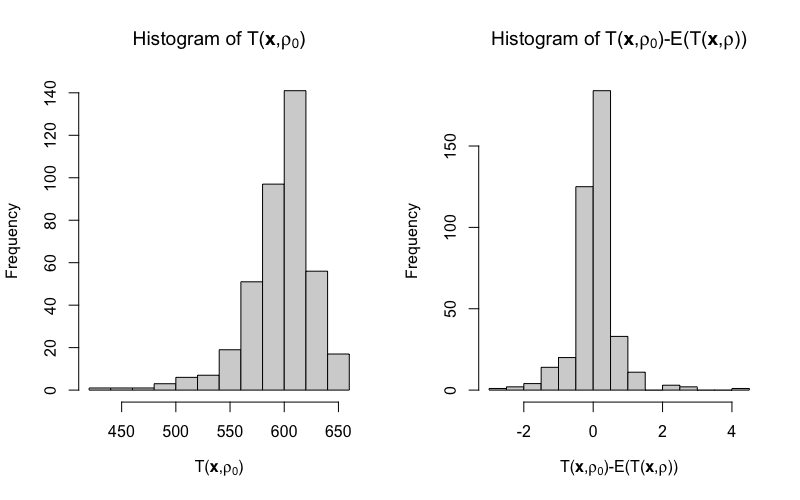}
\caption{Histogram of $T(\bm x,\rho_0)$ and the gap $T(\bm x,\rho_0)-\mathbb{E}[T(\bm x,\rho)]$}\label{fg:gap}
\end{figure}

In more general case, Proposition \ref{prop:gap} provides the analytic gap between $T(\bm x, \rho_0)$ and $\mathbb{E}[T(\bm x, \rho)]$.
Its proof is provided in the Supplement.
Proposition \ref{prop:gap} also provides two different upper bounds of the gap.
Which one of the two upper bounds is larger depends on the adjacency matrix $\bm W$ and $\rho_0$.
Regrettably, since both the upper bounds are independent of the design $\bm x$, they are too large to have any practical guidance, even though they might still be attainable for certain extreme design $\bm x$.
For the above simulation example, since the skewness of uniform distribution is 0, the two upper bounds of \eqref{eq:gap-upper} and \eqref{eq:gap-upper2} are calculated as 902.4 and 650.1, respectively.
They are much larger than the range shown in the histogram in Figure \ref{fg:gap}.
On the other hand, the two upper bounds increase as the size and density of the network become larger.
Therefore, for large and dense networks we should be more careful applying the locally optimal design.

\begin{proposition}\label{prop:gap}
The difference between $T(\bm x, \rho_0)$ and $\mathbb{E}(T(\bm x, \rho))$ is
\begin{equation}\label{eq:gap}
T(\bm x,\rho_0)-\mathbb{E}(T(\bm x,\rho))=\frac{1}{2}\left.\frac{\partial^2 T_2(\bm x, \rho)}{\partial \rho^2}\right\vert_{\rho=\rho_0}\var(\rho)-\mathbb{E}(O(\rho-\rho_0)^3),
\end{equation}
where
\begin{align}
\label{eq:deriv-2nd}
\frac{1}{2}\left.\frac{\partial^2 T_2(\bm x, \rho)}{\partial \rho^2}\right\vert_{\rho=\rho_0}&=\bm s^\top \bm W\bm F\left[\bm F^\top (\bm D-\rho_0 \bm W)\bm F\right]^{-1}\bm F^\top \bm W\bm s,\\
\label{eq:z}
\text{and}\quad \bm s&:=\left[\bm I_n-\bm F (\bm F^\top(\bm D-\rho_0\bm W)\bm F)^{-1}\bm F^\top (\bm D-\rho_0\bm W)\right]\bm x.
\end{align}
An upper bound of the gap $T(\bm x,\rho_0)-\mathbb{E}(T(\bm x,\rho))$ is
\begin{equation}\label{eq:gap-upper}
T(\bm x,\rho_0)-\mathbb{E}(T(\bm x,\rho))\leq \min\left\{n\lambda_{\max}(\bm D-\rho_0\bm W), (1+\rho_0)m\right\}\frac{|\lambda(\bm W)|_{\max}^2\var(\rho)}{\lambda_{\min}^2(\bm D-\rho_0\bm W)}-\mathbb{E}\left[O(\rho-\rho_0)^3\right],
\end{equation}
where $\lambda_{\min}(\bm D-\rho_0\bm W)$ and $\lambda_{\max}(\bm D-\rho_0\bm W)$ are the minimum and maximum eigenvalues of the Laplacian matrix $\bm D-\rho_0\bm W$, which is positive definite for $\rho_0\in (0, 1)$, $|\lambda(\bm W)|_{\max}$ is the spectrum radius of $\bm W$, and $m=\sum_{i=1}^n m_i$.
Based on Theorem \ref{thm:t1}, an alternative upper bound \eqref{eq:gap-upper2} holds asymptotically with probability of $100(1-\alpha)\%$ and $\alpha \in (0,1)$,
\begin{equation}\label{eq:gap-upper2}
T(\bm x,\rho_0)-\mathbb{E}(T(\bm x,\rho))\leq (m+z_{\alpha}\sqrt{m})\frac{|\lambda(\bm W)|_{\max}^2\var(\rho)}{\lambda_{\min}^2(\bm D-\rho_0\bm W)}-\mathbb{E}\left[O(\rho-\rho_0)^3\right],
\end{equation}
where $z_{\alpha}=\Phi^{-1}(\alpha)$ is the upper $\alpha$ quantile of the standard normal distribution.
\end{proposition}

\begin{lem}\label{lem:quad_eq1}
Let $\bm A$ be an $n\times n$ real symmetric positive definite matrix. For any vector $\bm x\in \mathbb{R}^n$, $\lambda_{\min}(\bm A)||\bm x||_2^2\leq \bm x^\top\bm A \bm x \leq \lambda_{\max}(\bm A)||\bm x||_2^2$. The equality holds if $\bm x={\bf 0}$ or $\bm A=a\bm I_n$ for $a\geq 0$.
\end{lem}
\begin{proof}
Because $\bm A$ is a real symmetric positive definite matrix, via eigendecomposition, $\bm A=\bm Q\bm \Lambda \bm Q^{-1}$, where $\bm \Lambda=\diag\{\lambda_1,\ldots, \lambda_n\}$ is a diagonal matrix of the eigenvalues of $\bm A$, $\bm Q$ is the square $n\times n$ matrix whose $i$th column is the eigenvector corresponding to eigenvalue $\lambda_i$.
Also, $\bm Q^\top =\bm Q^{-1}$.
Denote $\bm l:=\bm Q^\top \bm x$.
\begin{align*}
&\bm x^\top \bm A\bm x=\bm x^\top \bm Q\bm \Lambda \bm Q^\top \bm x=\bm l^\top \bm \Lambda \bm l=\sum_{i=1}^n \lambda_il_i^2, \\
\lambda_{\min}(\bm A)||\bm l||_2^2&=\lambda_{\min} (\bm A)\sum_{i=1}^n l_i^2 \leq \sum_{i=1}^n \lambda_i l_i^2\leq \lambda_{\max}(\bm A)\sum_{i=1}^n l_i^2=\lambda_{\max}(\bm A)||\bm l||_2^2.
\end{align*}
Here $\lambda_{\max}(\bm A)$ and $\lambda_{\min}(\bm A)$ are the maximum and minimum eigenvalues of $\bm A$, and since $\bm A$ is positive definite, $\lambda_{\min}(\bm A)>0$.
The norm $||\cdot||_2$ is the $l_2$-norm of a vector, and $||\bm l||_2^2=\bm l^\top \bm l=\bm x^\top \bm Q\bm Q^\top \bm x=||\bm x||_2^2$.
Thus the lemma is proved.
\end{proof}

\begin{lem}\label{lem:quad_eq2}
Let $\bm A$ be an $n\times n$ real symmetric matrix. For any vector $\bm x\in \mathbb{R}^n$, $|\bm x^\top \bm A \bm x| \leq |\lambda(\bm A)|_{\max}||\bm x||_2^2$.
\end{lem}
\begin{proof}
For any real symmetric matrix, based on eigenvalue decomposition, $\bm A=\bm Q\bm \Lambda \bm Q^\top $, where $\bm \Lambda=\diag\{\lambda_1,\ldots, \lambda_n\}$ is a diagonal matrix of the eigenvalues of $\bm A$, and $\bm Q$ is the $n\times n$ orthogonal matrix as above.
Denote $\bm l:=\bm Q^\top \bm x$.
\[
|\bm x^\top \bm A\bm x|=|\bm x^\top \bm Q\bm \Lambda \bm Q^\top \bm x|=|\bm l^\top \bm \Lambda \bm l|=|\sum_{i=1}^n \lambda_il_i^2|\leq \sum_{i=1}^n |\lambda_i|l_i^2\leq |\lambda(\bm A)|_{\max}||\bm l||_2^2=|\lambda(\bm A)|_{\max}||\bm x||_2^2.
\]
Here $|\lambda(\bm A)|_{\max}=\max_{i=1,\ldots,n}|\lambda|_i$.
\end{proof}

\noindent{\bf Proof of Proposition \ref{prop:gap}}
\begin{proof}
Using Taylor expansion, we have
\begin{align*}
T(\bm x, \rho)&=T(\bm x,\rho_0)+\left.\frac{\partial T(\bm x,\rho)}{\partial \rho}\right\vert_{\rho=\rho_0}(\rho-\rho_0)+\frac{1}{2}\left.\frac{\partial^2 T(\bm x,\rho)}{\partial \rho^2}\right\vert_{\rho=\rho_0}(\rho-\rho_0)^2+O((\rho-\rho_0)^3).
\end{align*}
Apply expectation on both side of the equaiton with respet the priori $p(\rho)$, we have
\begin{align*}
\mathbb{E}\left[T(\bm x, \rho)\right]&=T(\bm x,\rho_0)+\left.\frac{\partial T(\bm x,\rho)}{\partial \rho}\right\vert_{\rho=\rho_0}\mathbb{E}\left[\rho-\rho_0\right]+\frac{1}{2}\left.\frac{\partial^2 T(\bm x,\rho)}{\partial \rho^2}\right\vert_{\rho=\rho_0}\mathbb{E}\left[(\rho-\rho_0)^2\right]+\mathbb{E}\left[O((\rho-\rho_0)^3)\right]\\
&=T(\bm x,\rho_0)+\frac{1}{2}\left.\frac{\partial^2 T(\bm x,\rho)}{\partial \rho^2}\right\vert_{\rho=\rho_0}\var(\rho)+\mathbb{E}\left[O((\rho-\rho_0)^3)\right].
\end{align*}
From the proof of Theorem \ref{thm:concave}, we have that
\[
\frac{\partial^2 T(\bm x, \rho)}{\partial \rho^2}=-\frac{\partial^2 T_2(\bm x,\rho)}{\partial \rho^2}.
\]
Thus we obtain the gap between $T(\bm x, \rho_0)$ and $\mathbb{E}\left[T(\bm x, \rho)\right]$ in \eqref{eq:gap}.
Also in proof of Theorem \ref{thm:concave},
\[
\frac{1}{2}\left.\frac{\partial^2 T_2(\bm x, \rho)}{\partial \rho^2}\right\vert_{\rho=\rho_0}=\bm s^\top \bm W\bm F\bm A^{-1}\bm F^\top\bm W \bm s,
\]
where
\begin{align*}
\bm A&=\bm F^\top (\bm D-\rho_0\bm W)\bm F,\\
\bm s&=\left[\bm I_n-\bm F\bm A^{-1}\bm F^\top (\bm D-\rho_0\bm W)\right]\bm x.
\end{align*}
From the definition of $\bm s$, we can see that
\begin{align*}
\bm s^\top (\bm D-\rho_0\bm W)\bm s &=\bm x^\top \left[(\bm D-\rho_0 \bm W)-(\bm D-\rho_0\bm W)\bm F\bm A^{-1}\bm F^\top (\bm D-\rho_0\bm W)\right]\bm x\\
&\leq \bm x^\top (\bm D-\rho_0\bm W)\bm x.
\end{align*}
From Lemma \ref{lem:quad_eq1}, since $\bm D-\rho_0\bm W$ is a real symmetric positive definite matrix as $\rho_0\in (0,1)$,
\[
\lambda_{\min}(\bm D-\rho_0\bm W) ||\bm s||_2^2\leq \lambda_{\max}(\bm D-\rho_0\bm W)||\bm x||_2^2=\lambda_{\max}(\bm D-\rho_0\bm W)n.
\]
On the other hand, $\bm x^\top (\bm D-\rho_0\bm W)\bm x\leq (1+\rho_0)m$.
Thus,
\[
||\bm s||_2^2\leq \frac{1}{\lambda_{\min}(\bm D-\rho_0\bm W)}\min\{n\lambda_{\max}(\bm D-\rho_0\bm W), (1+\rho_0)m\}.
\]
According to Theorem \ref{thm:t1}, $\bm x^\top \bm W \bm x/\sqrt{m}$ converges in distribution to the standard normal distribution.
Therefore, with probability of $100(1-\alpha)\%$, $\bm x^\top \bm W\bm x \geq -z_{\alpha} \sqrt{m}$, asymptotically.
Here $z_{\alpha}$ is the upper $\alpha$ quantile of the standard normal distribution, i.e., $z_{\alpha}=\Phi^{-1}(1-\alpha)$.
So we can obtain an asymptotic upper bound,
\[
\bm s^\top (\bm D-\rho_0\bm W)\bm s\leq \bm x^\top (\bm D-\rho_0\bm W)\bm x=\bm x^\top\bm D\bm x-\rho_0\bm x^\top \bm W\bm x=m-\rho_0\bm x^\top \bm W\bm x\leq m+z_{\alpha}\sqrt{m},
\]
which holds with probability of $100(1-\alpha)\%$.
Consequently, an asymptotic upper bound for $||s||_2^2$ is
\[||\bm s||_2^2\leq \frac{1}{\lambda_{\min}(\bm D-\rho_0\bm W)} (m+z_{\alpha}\sqrt{m})\]
with probability of $100(1-\alpha)\%$.

It is easy to see that the matrix
\[
\bm I_n-(\bm D-\rho_0\bm W)^{1/2}\bm F\bm A^{-1}\bm F^\top (\bm D-\rho_0\bm W)^{1/2}
\]
is a projection matrix, and thus
\begin{align*}
&\bm s^\top \bm W\bm F\bm A^{-1}\bm F^\top\bm W\bm s \\
= & \bm s^\top \bm W(\bm D-\rho_0\bm W)^{-1/2}(\bm D-\rho_0\bm W)^{1/2}\bm F\bm A^{-1}\bm F^\top (\bm D-\rho_0\bm W)^{-1/2}(\bm D-\rho_0\bm W)^{1/2}\bm W\bm s\\
\leq & \bm s^\top \bm W(\bm D-\rho_0\bm W)^{-1}\bm W\bm s \leq \lambda_{\min}^{-1}(\bm D-\rho_0\bm W)||\bm W\bm s||_2^2\\
\leq &\lambda_{\min}^{-1}(\bm D-\rho_0\bm W)||\bm W||_2^2||\bm s||_2^2 = \lambda_{\min}^{-1}(\bm D-\rho_0\bm W)|\lambda(\bm W)|_{\max}^2||\bm s||_2^2
\end{align*}
The first inequality is due to Lemma \ref{lem:quad_eq2}.
Here $|\lambda(\bm W)|_{\max}=||\bm W||_2$ is the spetrum radius of $\bm W$.
Combining the previous steps we obtain the upper bound of the gap in \eqref{eq:gap-upper}.
\end{proof}

\section*{S3. Proposition \ref{prop:cor} and Its Proof}

\begin{proposition}\label{prop:cor}
Let $x_1, \ldots, x_n$ of $\bm x$ are independent and identically distributed random variables from the
 discrete distribution with $\Pr(x_i=1)=\Pr(x_i=-1)=0.5$. 
For any two symmetric and non-zero $n\times n$ matrices $\bm A$ and $\bm B$, we have that
\begin{equation}\label{eq:quadcor}
\cor_{\bm x}(\bm x^\top \bm A \bm x, \bm x^\top \bm B \bm x)=\frac{\sum_{i < j} a_{ij}b_{ij}}{\sqrt{\sum_{i < j}a^2_{ij}}\sqrt{\sum_{i < j} b^2_{ij}}},
\end{equation}
where $a_{ij}$ and $b_{ij}$ are the $(i,j)$-th entries of matrices $\bm A$ and $\bm B$ respectively.
\end{proposition}

Consider two $n\times n$ symmetric matrices $\bm A$ and $\bm B$. For random designs, we have that 
$\mathbb{E}(x_i)=0$, $\var(x_i)=1$, and $\cov(x_i,x_j)=0$ for $i\neq j$. Therefore, $\cov(\bm x)=\bm I_n$ and 
\begin{align*}
\cov(\bm x^\top \bm A \bm x, \bm x^\top \bm B \bm x )&=\mathbb{E}(\bm x^\top \bm A \bm x\bm x^\top \bm B \bm x)-
\mathbb{E}(\bm x^\top \bm A \bm x)\mathbb{E}(\bm x^\top \bm B \bm x)\\
&=\mathbb{E}\left(\bm x^\top \bm A \bm x\bm x^\top \bm B \bm x\right)-\tr(\bm A)\tr(\bm B)
\end{align*}
Note that 
\[
\bm x^\top \bm A \bm x\bm x^\top \bm B \bm x=(\bm x^\top \bm A \bm x)\otimes(\bm x^\top \bm B \bm x)=(\bm x^\top \otimes\bm x^\top)(\bm A\otimes \bm B) (\bm x \otimes \bm x).
\]
Then 
\begin{align*}
\bm x^\top \bm A \bm x\bm x^\top \bm B \bm x &= \tr(\bm x^\top \bm A \bm x\bm x^\top \bm B \bm x)=\tr((\bm x^\top \otimes\bm x^\top) (\bm A\otimes \bm B) (\bm x \otimes \bm x))\\
&=\tr((\bm A\otimes \bm B) (\bm x \otimes \bm x)(\bm x^\top \otimes\bm x^\top)),
\end{align*}
and thus
\[
\mathbb{E}(\bm x^\top \bm A \bm x\bm x^\top \bm B \bm x)=\mathbb{E}(\tr(\bm x^\top \bm A \bm x\bm x^\top \bm B \bm x))=\tr((\bm A\otimes \bm B)\mathbb{E}( (\bm x \otimes \bm x)(\bm x^\top \otimes\bm x^\top)))
\]
We need to derive $\mathbb{E}( (\bm x \otimes \bm x)(\bm x^\top \otimes\bm x^\top))$. 
Note that $(\bm x \otimes \bm x)(\bm x^\top \otimes\bm x^\top)=(\bm x\bm x^\top)\otimes (\bm x \bm x^\top)$
is an $n\times n$ block matrix, and the $i,j$-th block is $x_ix_j \bm x\bm x^\top$.
The diagonal blocks are $\mathbb{E}(x_i^2\bm x\bm x^\top)=\bm I_n$ ($\mathbb{E}(x_i^4)=1)$. 
If $i\neq j$,  $\mathbb{E}(x_ix_j \bm x\bm x^\top)=\bm e_i \bm e^\top_j+\bm e_j\bm e_i^\top$, where $\bm e_i$ is the element vector with $i$-th entry equal to 1 others 0 and $\bm e_i \bm e^\top_j+\bm e_j \bm e_i^\top$ is a matrix with $(i,j)$th and $(j,i)$th entries equal to 1 and the rest entries 0. 
Therefore, the resulting $n\times n$ block matrix should have diagonal blocks be an $n\times n$ identity matrix, and the $(i,j)$-th off-diagonal block be $\bm e_i \bm e^\top_j+\bm e_j\bm e_i^\top$. 
So we can decompose the block matrix to be 
\begin{align*}
&\mathbb{E}((\bm x \otimes \bm x)(\bm x^\top \otimes\bm x^\top))=\bm I_n\otimes \bm I_n+\sum_{i \neq j} (\bm e_i\bm e^\top_j)\otimes (\bm e_i\bm e^\top_j+\bm e_j\bm e^\top_i)\\
=& \bm I_n\otimes \bm I_n+\sum_{i \neq j} (\bm e_i\bm e^\top_j)\otimes (\bm e_i\bm e^\top_j)+\sum_{i \neq j} (\bm e_i\bm e^\top_j)\otimes(\bm e_j\bm e^\top_i)
\end{align*}

Then 
\begin{align*}
&\tr\left[(\bm A\otimes \bm B)\mathbb{E}[(\bm x \otimes \bm x)(\bm x^\top \otimes\bm x^\top)]\right]\\
=&\tr\left[(\bm A\otimes \bm B)(\bm I_n\otimes \bm I_n)\right]+\tr\left[\sum_{i \neq j} (\bm A\otimes\bm B)[(\bm e_i\bm e^\top_j)\otimes (\bm e_i\bm e^\top_j)]\right]+\tr\left[\sum_{i \neq j} (\bm A\otimes\bm B)[(\bm e_i\bm e^\top_j)\otimes (\bm e_j\bm e^\top_i)]
\right]\\
=& \tr\left[\bm A\otimes \bm B\right]+\sum_{i \neq j} \tr \left[(\bm A\otimes \bm B)[(\bm e_i\bm e_j^\top)\otimes (\bm e_i\bm e_j^\top)]\right]+\sum_{i \neq j} \tr \left[(\bm A\otimes \bm B)[(\bm e_i\bm e_j^\top)\otimes (\bm e_j\bm e_i^\top)]\right]\\
=&\tr (\bm A)\tr (\bm B)+\sum_{i \neq j} \tr \left[ (\bm A \bm e_i\bm e^\top_j)\otimes (\bm B \bm e_i\bm e^\top_j)\right]+\sum_{i \neq j} \tr \left[ (\bm A \bm e_i\bm e^\top_j)\otimes (\bm B \bm e_j\bm e^\top_i)\right]\\
=& \tr (\bm A)\tr (\bm B)+2\sum_{i \neq j} \tr \left[ (\bm A \bm e_i\bm e^\top_j) \right]\tr\left[(\bm B \bm e_i\bm e^\top_j)\right]\\
=&\tr (\bm A)\tr (\bm B)+4\sum_{i < j} \bm A_{ij}\bm B_{ij},
\end{align*}
where $\bm A_{ij}$ is the $ij$-th entry of matrix $\bm A$.
Then
\[
\cov(\bm x^\top \bm A \bm x, \bm x^\top \bm B \bm x )=4\sum_{i < j} \bm A_{ij}\bm B_{ij}
\]
Accordingly, 
\[
\cor(\bm x^\top \bm A \bm x, \bm x^\top \bm B \bm x)=\frac{\sum_{i < j} \bm A_{ij}\bm B_{ij}}{\sqrt{\sum_{i < j} \bm A^2_{ij}}\sqrt{\sum_{i < j} \bm B^2_{ij}}}
\]

\section*{S4. Proof of Theorem \ref{thm:t1}}
We first provide a useful Lemma.
\begin{lem}\label{lem}
Let $X$ and $Y$ be two random variables taking values from $\{-1, 1\}$.
If $\mathrm{cov}(X, Y)=0$, then $X$ and $Y$ are independent.

\end{lem}
\begin{proof}
Let $U$ and $V$ be two Bernoulli random variables. We first show that
if $\mathrm{cov}(U, V)=0$, then $U$ and $V$ are independent.

Notice that
\[
\Pr(\{U=1\}\mathrm{ and }\{V=1\})=\Pr(UV=1)=\mathbb{E}(UV)
\]
\[
\mathbb{E}(U)=\Pr(U=1)
\]
and
\[
\mathbb{E}(V)=\Pr(V=1).
\]
If $\mathrm{cov}(U, V)=0$,
\[
\Pr(\{U=1\}\mathrm{ and }\{V=1\})-\Pr(U=1)\Pr(V=1)
=\mathbb{E}(UV)-\mathbb{E}(U)\mathbb{E}(V)=0.
\]
Similarly, we can show that
\[
\Pr(\{U=0\}\mathrm{ and }\{V=1\})-\Pr(U=0)\Pr(V=1)=0,
\]
\[
\Pr(\{U=0\}\mathrm{ and }\{V=0\})-\Pr(U=0)\Pr(V=0)=0,
\]
and
\[
\Pr(\{U=1\}\mathrm{ and }\{V=0\})-\Pr(U=1)\Pr(V=0)=0,
\]
which demonstrate that $U$ and $V$ are independent.

For $X$ and $Y$, we have that $X=2U-1$ and $Y=2V-1$. The independence of $U$ and $V$ indicates the independence of $X$ and $Y$.
Also,
\[
\mathrm{cov}(X, Y)=4\mathrm{cov}(U, V).
\]
Thus, the conclusion holds.
\end{proof}

\begin{proof}
	Recall that $w_{ii}=0$ for $i=1,\ldots, n$. Therefore, we only need to consider the terms $w_{ij}x_ix_j$ with $i\neq j$.
Notice that
\[
\mathrm{cov}(x_ix_j, x_{i'}x_{j'})=\mathbb{E}(x_ix_jx_{i'}x_{j'})-\mathbb{E}(x_ix_j)\mathbb{E}(x_{i'}x_{j'})=0
\]
for $i\neq i'$ and $j\neq j'$. Also,
\[
\mathrm{cov}(x_ix_j, x_{i}x_{j'})=\mathbb{E}(x^2_ix_jx_{j'})-\mathbb{E}(x_ix_j)\mathbb{E}(x_{i}x_{j'})=0
\]
for $j\neq j'$. According to Lemma \ref{lem}, we have that $x_ix_j$ and $x_ix_{j'}$ are independent, and
 $x_ix_j$ and $x_{i'}x_{j'}$ are independent.
 Thus, $w_{ij}x_ix_j$'s with $w_{ij}\neq 0$ are i.i.d random variables with mean
 \[
 \mathbb{E}(w_{ij}x_ix_j)= \mathbb{E}(x_i) \mathbb{E}(x_j)=0,
 \]
 and variance
  \[
 \var(w_{ij}x_ix_j)= \mathbb{E}(x^2_i x^2_j)-(\mathbb{E}(x_i x_j))^2=1.
 \]
According to
the central limit theorem, the conclusion holds.
\end{proof}

\section*{S5. Proof of Proposition \ref{prop:prec}}

\begin{proof}

Notice that
\[
\mathbb{E}\left(\bm x^\top \bm K \bm x\right)=\tr\left[\mathbb{E}\left(\bm x^\top \bm K \bm x\right)\right]
=\mathbb{E}\left[\tr(\bm x^\top \bm K \bm x)\right]=\mathbb{E}\left[\tr(\bm K \bm x \bm x^\top )\right]=
\tr\left[\bm K \mathbb{E}(\bm x \bm x^\top)\right].
\]
For completely random design, under the same assumption as in Theorem \ref{thm:t1}, we have that
\[
\mathbb{E}(x_i x_j)=\mathbb{E}(x_i)\mathbb{E}(x_j)=0~~\mathrm{for}~~i\neq j
\]
and $\mathbb{E}(x^2_i)=1$ for $i=1, \ldots, n$.
Thus, $\mathbb{E}(\bm x \bm x^\top)=\bm I_n$.

Now we consider the case where $\bm x$ is a random balanced design. 
If $n$ is even, we have that
\[
\mathbb{E}\left(x_i\sum^n_{j=1} x_j\right)=0
\]
 since the balanced constraint gives $\sum^n_{j=1} x_j=0$ directly.
If $n$ is odd, $n=2h+1$ with $h$ be a positive integer. 
Due to the balance constraint, $\sum_{i=1}^n x_i=1$ or $-1$. 
We have that
\begin{align*}
&\mathbb{E}\left(x_i\sum^n_{j=1} x_j\right)=\Pr\left(\sum^n_{i=1} x_i=1\right)\mathbb{E}\left(x_i\sum^n_{j=1} x_j\bigg |\sum^n_{i=1} x_i=1\right)+\Pr\left(\sum^n_{i=1} x_i=-1\right)\mathbb{E}\left(x_i\sum^n_{j=1} x_j\bigg |\sum^n_{i=1} x_i=-1\right)\\
&=\frac{1}{2} \mathbb{E}\left(x_i\bigg |\sum^n_{j=1} x_j=1\right)-\frac{1}{2} \mathbb{E}\left(x_i\bigg |\sum^n_{j=1} x_j=-1\right).
\end{align*}
Note that 
\begin{align*}
&\mathbb{E}\left(x_i\bigg|\sum^n_{i=1} x_i=1\right)=\Pr\left(x_i=1\bigg |\sum^n_{i=1} x_i=1\right)-\Pr\left(x_i=-1\bigg |\sum^n_{i=1} x_i=1\right)=\frac{h+1}{2h+1}-\frac{h}{2h+1}=\frac{1}{n},\\
&\mathbb{E}\left(x_i\bigg|\sum^n_{j=1} x_j=-1\right)=\Pr\left(x_i=1\bigg |\sum^n_{j=1} x_j=-1\right)-\Pr\left(x_i=-1\bigg |\sum^n_{j=1} x_j=-1\right)=\frac{h}{2h+1}-\frac{h+1}{2h+1}=-\frac{1}{n}.
\end{align*}
Thus, $\mathbb E\left(x_i\sum^n_{j=1} x_j\right)=1/n$.

Therefore,
\[
\mathbb E\left(x_1\sum^n_{j=1} x_j\right)=1+(n-1)\mathbb E(x_1x_2)\]
which gives that 
\[
\mathbb E(x_1x_2)=\begin{cases}
-\frac{1}{n-1} ~~\mathrm{if}~~n~\mathrm{is~ even}\\
-\frac{1}{n}~~if~~\mathrm{if}~~n~\mathrm{is~ odd}
\end{cases}.
\]
This conclusion holds for $\mathbb{E}(x_ix_j)$ with any $i\neq j$. 
\end{proof}


\end{document}